\documentclass[a4paper,10pt,onecolumn,
accepted=2026-04-28]{quantumarticle}
\newif\ifabstract % Extended abstract
\abstractfalse

\newif\ifdraft
\draftfalse

\pdfoutput=1

\usepackage[marginpar=2cm,margin=2cm]{geometry}
\usepackage[utf8]{inputenc}
\usepackage[T1]{fontenc}
\usepackage[english]{babel}

\usepackage[svgnames]{xcolor}

\usepackage{float}

\usepackage{tikz}
\usepackage{pgfplots}
\usetikzlibrary{arrows.meta, calc}
\pgfplotsset{compat = 1.17}

\usepackage{amsfonts}
\usepackage{amsmath}
\usepackage{amssymb}
\usepackage{amsthm}
\allowdisplaybreaks[2]

\usepackage{mathtools}

\usepackage{bbm}

\usepackage[figurewithin=none,margin=1cm,font=normalfont,labelfont=bf]{caption}
\usepackage[figurewithin=none]{subcaption}
\usepackage{xfrac}

\usepackage{authblk}

\usepackage{calc}

\usepackage{enumitem}

\usepackage[linesnumbered,algoruled]{algorithm2e}

\usepackage{settings}
\usepackage{hyperref}
\hypersetup{colorlinks=true,
            breaklinks=true,
            linkcolor = quantumviolet, 
            urlcolor = quantumviolet,
            citecolor = quantumviolet,
            pdfpagemode=UseOutlines}
\usepackage[nameinlink]{cleveref}
\crefname{conjecture}{conjecture}{conjectures}
\Crefname{Conjecture}{Conjecture}{Conjectures}

\newtheoremstyle{mystyle}% name of the style to be used
  {15pt}% measure of space to leave above the theorem. E.g.: 3pt
  {15pt}% measure of space to leave below the theorem. E.g.: 3pt
  {\it}% name of font to use in the body of the theorem
  {0.cm}% measure of space to indent
  {\bf}% name of head font
  {.}% punctuation between head and body
  {0.3cm}% space after theorem head; " " = normal interword space
  {}% Manually specify head

\newtheoremstyle{mystyle-not-italic}% name of the style to be used
  {15pt}% measure of space to leave above the theorem. E.g.: 3pt
  {15pt}% measure of space to leave below the theorem. E.g.: 3pt
  {}% name of font to use in the body of the theorem
  {0.cm}% measure of space to indent
  {\bf}% name of head font
  {.}% punctuation between head and body
  {0.3cm}% space after theorem head; " " = normal interword space
  {}% Manually specify head

\theoremstyle{mystyle}
\newtheorem{theorem}{Theorem}
\newtheorem*{theorem*}{Theorem}
\newtheorem{proposition}[theorem]{Proposition}
\newtheorem*{proposition*}{Proposition}

\newtheorem*{corollary*}{Corollary}
\newtheorem{lemma}[theorem]{Lemma}
\newtheorem*{lemma*}{Lemma}
\newtheorem{conjecture}{Conjecture}
\newtheorem*{conjecture*}{Conjecture}
\newtheorem{result}{Result}
\newtheorem*{result*}{Result}
\theoremstyle{mystyle-not-italic}
\newtheorem{definition}{Definition}
\newtheorem*{definition*}{Definition}

\newtheorem*{example*}{Example}
\theoremstyle{mystyle-not-italic}
\newtheorem{remark}{Remark}
\newtheorem*{remark*}{Remark}

\newcommand{\myorcid}[1]{\href{https://orcid.org/#1}{\includegraphics[height=0.25cm]{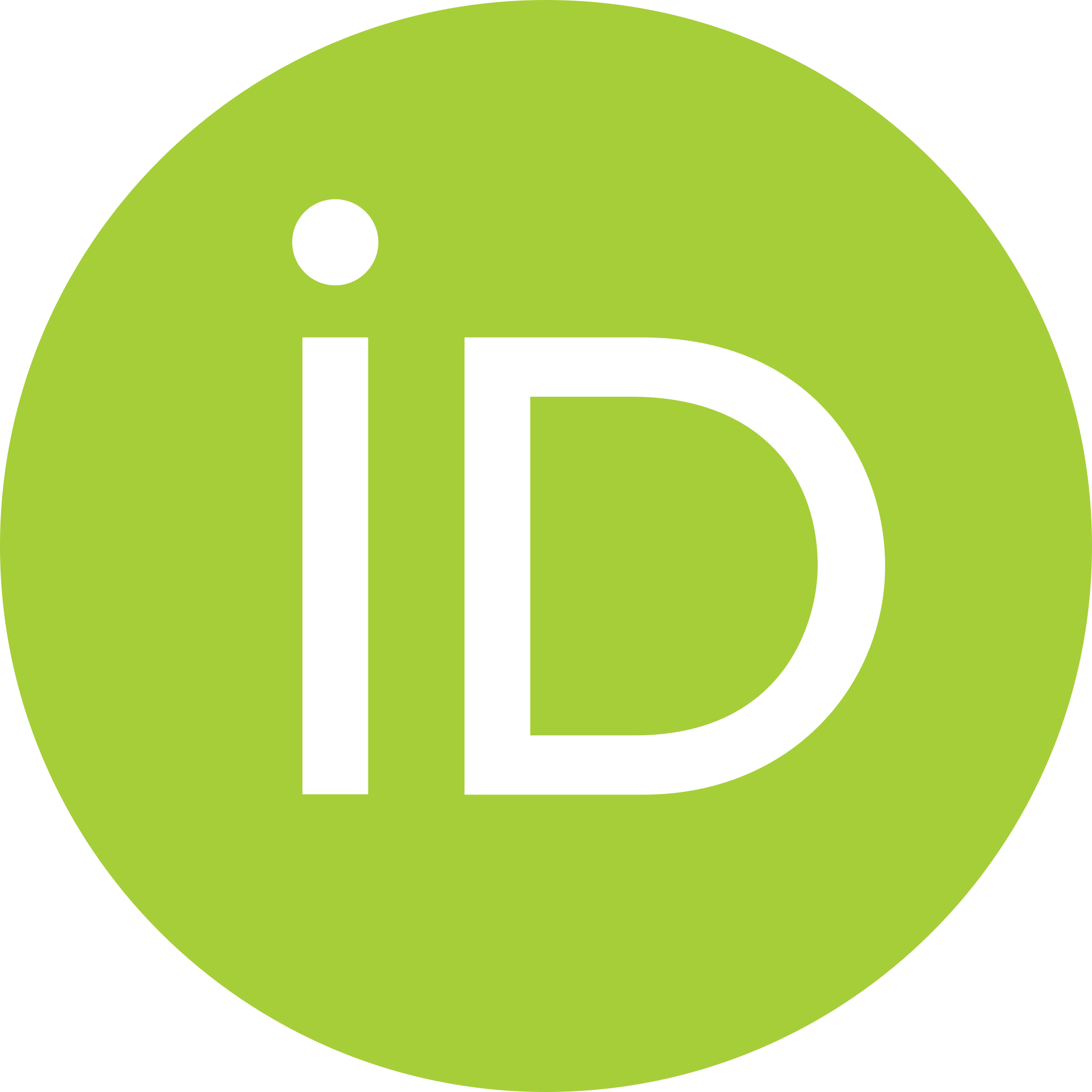}}}

\title{Towards Unconditional Uncloneable Encryption}
\author[1]{Pierre Botteron~\myorcid{0000-0002-3861-0934}}
\author[2]{Anne Broadbent~\myorcid{0000-0003-1911-0093}}
\author[3]{Eric Culf}
\author[4]{Ion Nechita~\myorcid{0000-0003-3016-7795}}
\author[1]{Clément Pellegrini}
\author[2]{Denis Rochette~\myorcid{0000-0001-5371-0426}}
\affil[1]{Institut de Mathématiques de Toulouse, Université de Toulouse, UPS, France.}
\affil[2]{Department of Mathematics and Statistics, University of Ottawa, Ottawa, Canada.}
\affil[3]{Institute for Quantum Computing, University of Waterloo, Waterloo, Canada.}
\affil[4]{Laboratoire de Physique Théorique, Université de Toulouse, UPS, France.}

\begin{document}

\fontfamily{bch}\selectfont % Beautiful font

\maketitle

\begin{abstract}
    Uncloneable encryption is a cryptographic primitive which encrypts a classical message into a quantum ciphertext, 
    such that two  quantum adversaries are limited in their capacity of being able to simultaneously decrypt,  
    given the key and quantum side-information produced from the ciphertext.  
    Since its initial proposal and scheme in the random oracle model by Broadbent and Lord [TQC 2020], 
    uncloneable encryption has developed into an important primitive at the foundation of quantum uncloneability for cryptographic primitives. 
    Despite sustained efforts, however, 
    the question of unconditional uncloneable encryption (and in particular of the simplest case, called an ``uncloneable bit'') has remained elusive. 
    Here, we propose a candidate for the unconditional uncloneable bit problem, and provide strong evidence that the adversary's success probability in the related security game converges quadratically as $\frac{1}{2}+\frac{1}{2\sqrt{K}}$, 
    where $K$ is polynomial in the size of the encoding, representing the number of keys and $\frac{1}{2}$ is trivially achievable. 
    We prove this bound's validity for $K$ ranging from $2$ to $7$ and  demonstrate the validity up to $K = 17$ using computations based on the NPA hierarchy. 
    We furthermore provide compelling heuristic evidence towards the general case. 
    In addition, we prove an asymptotic upper-bound of $\frac{5}{8}$ and give a numerical upper-bound of $\sim\!0.5980$, which to our knowledge is the best-known value in the unconditional~model.\looseness=-1
\end{abstract}

\vspace{-1cm}
\setcounter{tocdepth}{2}
\tableofcontents

\ifabstract
\section*{Introduction}
\else
\section{Introduction}
\fi
\label{sec:introduction}

Uncloneable encryption is an encryption scheme for classical messages into quantum ciphertexts that inherently protects the underlying plaintext from cloning, as defined by a security game that pits a triple of adversaries ($\Pirate$, $\Bob$, $\Charlie$) against an honest sender $\Alice$ (see \Cref{fig:uncloneable_encryption}).  Clearly, such unclonability is unattainable with classical ciphertexts since $\Pirate$ can send two copies of the  classical ciphertext and, given the decryption key, $\Bob$ and $\Charlie$ can  decrypt in a straightforward way. 

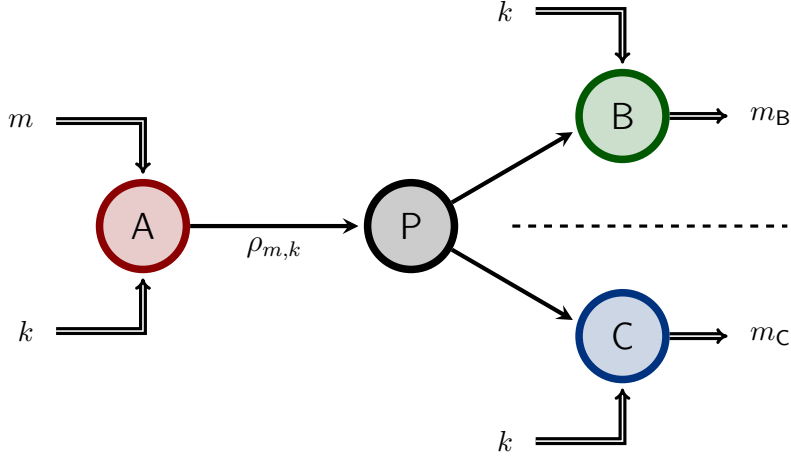
\begin{figure}
    \centering
    \ifabstract
        \resizebox{!}{6.5cm}{
    \else
    \fi
    \begin{tikzpicture}
        \newcommand{\fillproportion}{20}
        \tikzset{
            block/.style={draw, minimum width=3em, minimum height=3em, font=\sffamily\Large, circle, line width=1mm},
            arrow/.style={-stealth, ultra thick},
            double_arrow/.style={-{Implies[]}, very thick, double},
        }
    
        \node[block, draw=DarkRed, fill=DarkRed!\fillproportion] (A) {A};
        \node[block, draw=Black, fill=Black!\fillproportion, right=6em of A] (P) {P};
        \node[block, draw=DarkGreen, fill=DarkGreen!\fillproportion, above right=1.5em and 5em of P] (B) {B};
        \node[block, draw=DarkBlue, fill=DarkBlue!\fillproportion, below right=1.5em and 5em of P] (C) {C};

        \draw[dashed, very thick] ([xshift=3.5em] P.center) -- ([xshift=13em] P.center);

        \draw[arrow] (A.east) -- node[below, font=\large] {$\rho_{m,k}$} ([xshift=-2pt] P.west);
        
        \draw[arrow] (P.30) -- ([xshift=-2pt] B.200);
        \draw[arrow] (P.330) -- ([xshift=-2pt] C.160);
    
        \draw[double_arrow] ([xshift=-3em,yshift=2em] A.north) node [label = left:{$m$}] {} -- ([yshift=2em] A.north) -- ([yshift=2pt] A.north);
        \draw[double_arrow] ([xshift=-3em,yshift=-2em] A.south) node [label = left:{$k$}] {} -- ([yshift=-2em] A.south) -- ([yshift=-2pt] A.south);
        
        \draw[double_arrow] ([xshift=-3em,yshift=2em] B.north) node [label = left:{$k$}] {} -- ([yshift=2em] B.north) -- ([yshift=2pt] B.north);
        \draw[double_arrow] (B.east) -- ([xshift=2em] B.east) node [label = right:{$m_{\Bob}$}] {};
        
        \draw[double_arrow] ([xshift=-3em,yshift=-2em] C.south) node [label = left:{$k$}] {} -- ([yshift=-2em] C.south) -- ([yshift=-2pt] C.south);
        \draw[double_arrow] (C.east) -- ([xshift=2em] C.east) node [label = right:{$m_{\Charlie}$}] {};
    \end{tikzpicture}
    \ifabstract
        }
    \else
    \fi
    \caption{
        \emph{No-Cloning Game for a 1-Bit Message.} 
        Alice~($\Alice$)  encrypts a uniformly random message $m \in \{0,1\}$ using a classical key $k\in\set{1,..,K}$ into a quantum state $\rho_{m,k}\in\B(\Hilbert_\Alice)$. 
        She transmits it to a pirate~($\Pirate$) modeled by a \CPTP~map $\Phi: \B(\Hilbert_\Alice) \to \B(\Hilbert_{\Bob} \otimes \Hilbert_{\Charlie})$. 
        Bob~($\Bob$) and Charlie~($\Charlie$) are then given their respective registers $\Hilbert_{\Bob}$ and $\Hilbert_{\Charlie}$, as well as a copy of the key~$k$. 
        They  output $m_{\Bob}$, $m_{\Charlie} \in \set{0, 1}$, respectively, and \emph{collaboratively win} if and only if $m_\Bob\!=\!m_\Charlie\!=\!m$. 
        \emph{Uncloneable-indistinguishable security} holds if the winning probability is upper-bounded by $\sfrac{1}{2} + \negl(\lambda)$ for some security parameter~$\lambda$.
    }
    \label{fig:uncloneable_encryption}
\end{figure}

Originally proposed by Broadbent and Lord~\cite{BL20} who showed  its achievability in the quantum random oracle model, (and under a less stringent definition called \emph{uncloneable security}), uncloneable encryption has become an important building block for quantum cryptography, including for private-key quantum money~\cite{BL20}, preventing storage attacks~\cite{BL20}, quantum functional encryption~\cite{MM24arxiv}, quantum copy-protection~\cite{AK21}, quantum position verification~\cite{GALC25}, and uncloneable decryption~\cite{GZ20eprint, SW22arxiv, KT25}.   

Given the importance of uncloneable encryption, efforts have focused on its achievability under various models and definitions, including achievability 
 in the quantum random oracle model (QROM)~\cite{BL20, AKL+22, AKL23}, 
in an interactive version of the scenario~\cite{BC23arxiv}, 
in a device-independent variant with variable keys~\cite{KT25}, 
assuming the existence of specific types of obfuscation~\cite{AB24,CHV24arxiv}, and in a variant with quantum keys~\cite{AKY25}.

Many open questions remain in the study of uncloneable cryptography, notably the achievability of a scheme that provides \emph{uncloneable-indistinguishability} security in the sense originally defined in~\cite{BL20}: the security definition considers a game of the form of \Cref{fig:uncloneable_encryption}, but where a message 
$m\in \{0,1\}^n$ is selected by the adversary, and  the challenge that the adversaries $\Bob$ and $\Charlie$ face is to identify if the original message $m$, or a fixed message $0^n$, was encrypted, where the two cases happen with equal probability. In this scenario,  limitations on possible schemes have been identified ~\cite{MST21arxiv,AKL+22}. Notably, however, achievability in the standard model, even with computational assumptions, is wide open; for further discussion and a candidate scheme, see~\cite{CHV24arxiv}. 

At the heart of this intriguing open question is the simplest case, called the \emph{uncloneable bit}, where $m \in \{0,1\}$ (see \Cref{fig:uncloneable_encryption}), which, despite its simplicity, has remained unsolved in the plain model. Its importance is highlighted in~\cite{HKNY24}, where it is shown that a scheme for an uncloneable bit can be transformed into a scheme that encrypts general messages and that satisfies uncloneable-indistinguishability\footnote{For conventional encryption, encrypting a message bitwise with a single-bit encryption scheme typically yields a secure encryption; however, such a construction is not secure in the context of uncloneable encryption.}. 

\ifabstract
\subsection*{Results}
\else
\subsection{Results}
\fi

Our work focuses on the achievability of an encryption scheme that realizes an \emph{uncloneable bit} in the statistical model, thus without any computational or setup assumptions. However, given the apparent difficulty of achieving this task, we  relax the security requirement, and we ask that the success probability of the adversaries be no more than $\sfrac{1}{2} + f(\lambda)$, for some $f: \RR \to \RR$ such that $\lim_\lambda f(\lambda) = 0$ (weak security). 
This is a relaxation of the usual requirement that $f$ be a negligible function (strong security). 

Our contribution is a new candidate scheme for this relaxation of the uncloneable bit question. We prove weak security for some small security parameters and provide strong numerical evidence that it also holds asymptotically. 

\paragraph*{Candidate Scheme for an Uncloneable Bit.} Our candidate scheme is based on a family of pairwise anti-commuting $n$-qubit Pauli strings, of the following form for $n$ even:
\begin{equation}  \label{eq:Pauli_strings}
       X^{\otimes (i-1)} \otimes Y \otimes I^{\otimes (n-i)} \quad \text{ and } \quad X^{\otimes (i-1)} \otimes Z \otimes I^{\otimes (n-i)}, \qquad i \in \set{1, \ldots, n}\,.
\end{equation}
Note that there are $2n$ such strings. When $n$ is odd, we add to the above set  $X^{\otimes n}$ and obtain $2n+1$ strings.  
We index these strings as $\Gamma_k$, and based on this, define the following candidate (see details in\ifabstract{ the main manuscript}\else{~\Cref{sec:scheme}}\fi).

\begin{definition*}[Candidate Scheme for an Uncloneable Bit]
    Consider $\Gamma_1, ..., \Gamma_K$ Hermitian unitaries that anti-commute (\eg the pairwise anti-commuting Pauli strings in \cref{eq:Pauli_strings}), of dimension $d=2^\lambda$, with $K = 2\lambda$ for even\footnote{We only consider the case $\lambda$ even here and let the reader see the complete definition of the scheme in\ifabstract{ the main manuscript}\else{~\Cref{sec:Clifford_Algebra}}\fi.} $\lambda$.
    Sample uniformly at random a key $k\in\set{1,\dots,K}$.
    From this key~$k$, encrypt a classical message $m\in\set{0,1}$ into the following quantum state:
    $$
        \rho_{m,k} = \frac{2}{d}\, \frac{I_d + (\shortminus 1)^m \Gamma_k}{2}\,,
    $$ 
    which is the normalized projector onto the eigenvalue $(\shortminus 1)^m$ of $\Gamma_k$.
    For the decryption scheme of a quantum state $\rho_{m,k}$, measure it in the eigenbasis of $\Gamma_k$. The outcome is the decrypted message.
\end{definition*}

One can easily check the correctness of this protocol, \ie decrypting an encrypted message recovers the initial message with full probability. 

Our scheme can be seen as a generalization of the basic uncloneable encryption scheme \cite{BL20} that uses as key $\theta \in \{0,1\}$ and encodes a bit $m$ as $H^\theta\ket{m}$. This encoding corresponds to the $K=2$ case of our candidate. Already in the case $K=2$, using the proof techniques from monogamy-of-entanglement games~\cite{TFKW13}, the success probability of the adversary at the no-cloning game is $\frac{1}{2} + \frac{1}{2\sqrt{2}}$, which is consistent with the following conjecture (see\ifabstract{ the main manuscript}\else{~\Cref{conj:main_conjecture}}\fi). 

\begin{conjecture*}
    Our candidate encryption scheme is weakly uncloneable-indistinguishable secure with the following upper-bound:
    \[
        \Pr\of[\Big]{\text{$(\Pirate, \Bob, \Charlie)$ win the no-cloning game}}
        \,\leq\,
        \frac{1}{2}+\frac{1}{2\sqrt{K}}\,,
    \]
    where $K$ is the number of keys.
\end{conjecture*}

The contributions of this paper focus on this conjecture and are divided into several analytical and numerical results, as summarized in~\Cref{fig:results}. The source code for our numerical methods is available on GitHub\footnote{\url{https://github.com/denis-rochette/Towards-Unconditional-Uncloneable-Encryption}}.
Our methods are mainly based on the Navascués-Pironio-Acín (NPA) hierarchy~\cite{NPA08} and Sum-of-Squares (SoS) decompositions~\cite{McC01,Hel02}, which are standard primal/dual SDP algorithms to approximate the set of quantum (commuting) correlations.
Those methods provide upper-bounds on the winning probability:

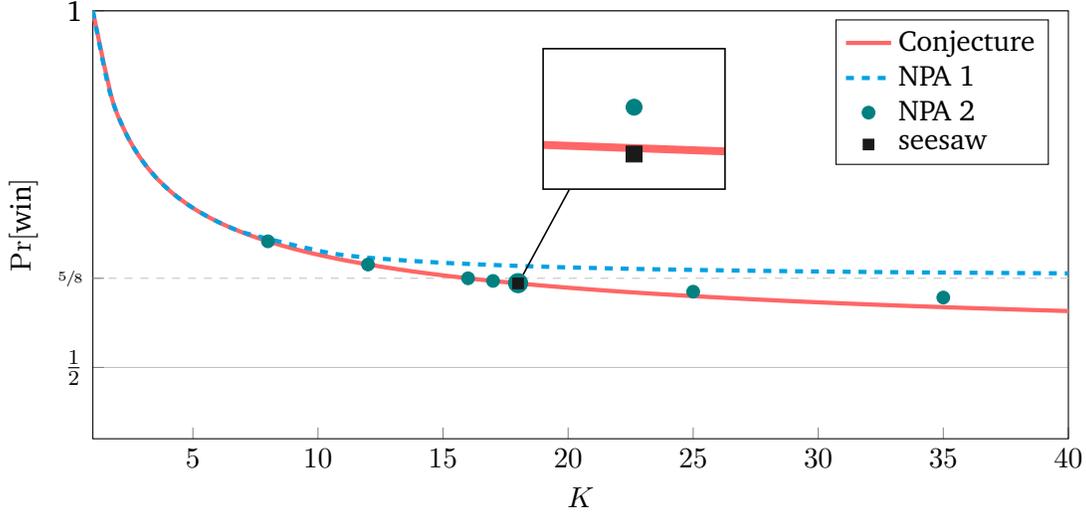
\begin{figure}[t]
    \centering
    \ifabstract
        \resizebox{!}{6.5cm}{
    \else
    \fi
    \begin{tikzpicture}
        \begin{axis}[
            xlabel = {$K$},
            ylabel = {$\Pr[\text{win}]$},
            ytick = {0, 1/2, 1},
            extra y ticks = {5/8},
            yticklabels = {0, $\frac{1}{2}$, 1},
            extra y tick labels = {$\sfrac{5}{8}$},
            extra y tick style = {major grid style=dashed, font=\tiny},
            xtick pos = bottom,
            ytick pos = left,
            xmin = 1,
            ymin = 0.4,
            xmax = 40,
            ymax = 1,
            ymajorgrids = true,
            major grid style = {line width=.2pt, draw=gray!50},
            legend cell align = left,
            width = 2*\axisdefaultheight,
            height = \axisdefaultheight
        ]
            
            \addplot[domain = {1:40},
                     samples = 50, 
                     color = red!60!white,
                     smooth,
                     ultra thick]
                     {1/4 + 1/(4*x) * (x + 2*sqrt(x))};
            \addlegendentry{Conjecture}
            
            \addplot[dashed,
                     domain = {1:7},
                     samples = 10, 
                     color = cyan!90!blue,
                     smooth,
                     ultra thick]
                     {1/2 + 1/(2*sqrt(x))};
            \addplot[dashed,
                     domain = {7:40},
                     samples = 10, 
                     color = cyan!90!blue,
                     smooth,
                     ultra thick,
                     forget plot]
                     {5/8 + 1/(2*(x-2)) - 1/(4*x)};
            \addlegendentry{NPA 1}
            
            \addplot+[mark=*,
                      only marks,
                      mark options={
                        scale=1.2,
                        fill=teal,
                        draw=teal,
                      }] 
	                  coordinates {
            		    (8,0.6768)
            		    (12,0.6443)
                        (16,0.6250)
                        (17,0.6213)
                        (25,0.6062)
                        (35,0.5980)
                      };
            \addplot+[mark=*,
                      only marks,
                      mark options={
                        scale=1.8,
                        fill=teal,
                        draw=teal,
                      },
                      forget plot] 
	                  coordinates {
                        (18,0.6182)
                      };
            \addlegendentry{NPA 2}

            \addplot+[mark=square*,
                      only marks,
                      mark options={
                        scale=1,
                        fill=black!90!white,
                        draw=black!90!white,
                      }] 
	                  coordinates {
                        (18,0.6178)
                      };
            \addlegendentry{seesaw}

            \coordinate (insetStart) at (18.2, {0.63});
            \coordinate (insetEnd) at (19, 0.75);
            \draw[-, semithick] (insetStart.center) -- ([xshift=10]insetEnd);
        \end{axis}
        \begin{axis}[
            tiny,
            at={(insetEnd)},
            anchor = south west,
            axis background/.style = {fill=white},
            xmin = 17.991,
            xmax = 18.009,
            ymin = 0.6175,
            ymax = 0.6187,
            axis on top,
            axis line style = semithick,
            ticks = none,
        ]
            \addplot[domain = {17.8:18.2},
                     samples = 4, 
                     color = red!60!white,
                     smooth,
                     line width = 3pt]
                     {1/4 + 1/(4*x) * (x + 2*sqrt(x))};
                     
            \addplot+[mark=*,
                      only marks,
                      mark options={
                        scale=1.5,
                        fill=teal,
                        draw=teal,
                      },
                      forget plot] 
	                  coordinates {
                        (18,0.6182)
                      };
                      
            \addplot+[mark=square*,
                      only marks,
                      mark options={
                        scale=1.5,
                        fill=black!90!white,
                        draw=black!90!white,
                      }] 
	                  coordinates {
                        (18,0.6178)
                      };
        \end{axis}
    \end{tikzpicture}
    \ifabstract
        }
    \else
    \fi
    \caption{Upper-bounds on the winning probability in the no-cloning game involving three adversaries $(\Pirate, \Bob, \Charlie)$ for our candidate scheme for Uncloneable Encryption with $K$ keys. 
    The solid line (red) is the conjectured upper-bound, 
    the dashed line (cyan) corresponds to the upper-bound derived from NPA level~1, 
    the circles (teal) are the numerical upper-bounds obtained from NPA level~2, 
    and the square (black) is the numerical result obtained using the seesaw optimization method on $K=18$.}
    \label{fig:results}
\end{figure}

\begin{result}[Cyan in \Cref*{fig:results}]
    The conjecture is confirmed for $K \leq 7$ using both an explicit Sum-of-Squares decomposition and an analytical level-1 NPA proof.
\end{result}

\begin{result}[Teal in \Cref*{fig:results}]
    The conjecture is numerically confirmed for $K \leq 17$ using a level-2 NPA optimisation\footnote{For larger~$K$, we would need the next levels of the NPA hierarchy---computational limitations prevented us from extending the numerics.}.
\end{result}

Then, using the analytic solution of the first level of the NPA hierarchy and taking its limit in~$K$, we also derive the following upper-bound:

\begin{result}[Cyan in \Cref*{fig:results}]
    Asymptotically, the winning probability of our candidate scheme for the no-cloning game is upper-bounded by $\sfrac{5}{8}$.
\end{result}

Additionally, we use an alternating optimization problem algorithm, known as seesaw methods, to obtain lower-bounds. We explore the case $K=18$ and found no instances that violate the conjecture:

\begin{result}[Black in \Cref*{fig:results}]
    For $K=18$, no violation of the conjecture is found using seesaw methods.
\end{result}

Finally, although we are unable to prove full uncloneable-indistinguishability, we prove that conventional indistinguishability also holds, in a relaxed sense, for our candidate scheme (see\ifabstract{ the main manuscript}\else{~\Cref{sec:indistinguishability}}\fi).

\ifabstract
\else

\subsection{Future Work}

Our work introduces a candidate scheme supported by strong evidence, highlighting its significance in the field of unconditional uncloneable encryption.
However, several important questions remain unresolved, which we outline below:

\paragraph*{Question 1: Proving the Conjecture for Large $K$.} In this paper, we provide an analytical proof for the conjecture when $K \leq 7$ and support it numerically for $K \leq 17$. Extending this proof to arbitrarily large values of $K$ would be a significant advancement, as it would mark the first unconditional proof of uncloneable encryption in the plain model.

\paragraph*{Question 2: Achieving Stronger Security.} \!Our proposed protocol achieves a quad\-ra\-tic security rate while the original definition in~\cite{BL20} specifies a security with exponential decay in the security parameter. A key open challenge is to develop a protocol that achieves this stronger security.

\subsection{Organization of the Paper}

The remainder of this paper is organized as follows. \Cref{sec:preliminaries} contains the necessary notation and preliminaries concerning states, measurements, and channels in quantum mechanics, as well as the formal definition of uncloneable encryption and its relation to monogamy-of-entanglement games. In~\Cref{sec:scheme}, we present our scheme for uncloneable encryption and establish a conjecture regarding its uncloneable-indistinguishable security based on  operator theory.
Finally, in~\Cref{sec:analytical-and-numerical-results}, we propose analytical proofs and numerical evidence of the conjecture in the first cases.
\fi
\section{Preliminaries} \label{sec:preliminaries}

We use $\NN_+$ to represent the set of positive natural numbers $\NN \setminus \set{0}$, and $[n]$ to denote the set $\set{1, \ldots, n}$. The indicator function is written as $\mathbbm{1}_{\set{\cdot}}$, and $\negl(\lambda)$ refers to a negligible function in $\lambda \in \NN_+$ \footnote{$ \forall n \in \NN, \exists N \in \NN$ such that $\forall \lambda > N$ we have $\negl(\lambda) \leq \lambda^{\shortminus n}$~\cite{KL07}.}.

In this paper, all Hilbert spaces are considered to be finite-dimensional. Given a Hilbert space $\Hilbert \simeq \CC^d$, we use $\M_d \simeq \B(\Hilbert)$ for the set of matrices (or operators) acting on $\Hilbert$, and $I_d$ for the identity matrix on this space. The trace on $\Hilbert$ is written as $\Tr[\cdot]$. An operator $M$ is positive semi-definite if and only if it is Hermitian (\ie $M^* = M$) with (real) non-negative eigenvalues. In this case, we write $M \succeq 0$. Given two Hermitian operators $A$ and $B$, we write $A \succeq B$ if $A - B \succeq 0$. The norm on $\Hilbert$ is denoted by $\norm{\cdot}$, and the operator norm on $\M_d$ by $\operatornorm{\cdot}$. Given two operators $A$ and $B$ acting on $\Hilbert$, the notation $\commutator{A}{B}$ represents the commutator $AB - BA$, and $\anticommutator{A}{B}$ represents the anti-commutator $AB + BA$. The Frobenius inner product of two operators $A$ and $B$ acting on $\Hilbert$ is defined by $\langle A, B \rangle \coloneqq \Tr[A^*B]$.

A Hilbert space $\Hilbert \simeq \CC^d$ is called a \emph{quantum system}, and its density matrices (\ie positive semi-definite unit-trace matrices) are the \emph{quantum states}, forming the following set:
\begin{equation*}
    \D_d \coloneqq \big\{ \rho \in \M_d \;:\; \Tr[\rho] = 1 \text{ and } \rho \succeq 0 \big\}\,.
\end{equation*}
The set of density matrices $\D_d$ is convex, and its extreme points are precisely the unit-rank projectors $\ketbra{\psi}{\psi}$, for some unit vector $\ket{\psi} \in \Hilbert$, referred to as \emph{pure quantum states}. On a composite system $\Hilbert^{\otimes n}$ with $n$ parties, we write $\Tr_i$ for the \emph{partial trace} on the Hilbert space corresponding to the $i$-th subsystem:
\begin{equation*}
    \Tr_i \coloneqq I_d \otimes \cdots \otimes I_d \otimes \underbracket[0.75pt]{\,\Tr\,}_{i\text{-th}} \otimes I_d \otimes \cdots \otimes I_d\,.
\end{equation*}
In the case of a composite system with different dimensions, we also use partial trace notation, indexed by the dimension to trace, \eg on $\Hilbert \simeq \CC^d \otimes \CC^D$, the notation $\Tr_D$ stands for the partial trace on the second tensor factor.
A \emph{quantum channel} $\Phi: \M_d \to \M_D$ is a Completely Positive and Trace Preserving (\CPTP) linear map. We use the notation $\Phi_i$ to denote the $i$-th marginal of the quantum channel: $\Phi_i \coloneqq \Tr_i \circ\, \Phi$. A Positive Operator-Valued Measure (\POVM) on a Hilbert space $\Hilbert \simeq \CC^d$ is a finite set of positive semi-definite operators $\set{M_i}$ that sum to the identity:~$\sum_i M_i = I_d$. A Projection-Valued Measure (\PVM) is a specific type of \POVM where $\set{M_i}$ is a set of orthogonal projections.

In this paper, we use the standard \emph{Pauli matrices}, which are Hermitian and unitary operators that serve as a basis for the space of $2 \times 2$ Hermitian matrices. They are defined as:
\begin{equation*}
    X \coloneqq \begin{bmatrix} 0 & 1 \\ 1 & 0 \end{bmatrix}, \quad
    Y \coloneqq \begin{bmatrix} 0 & -i \\ i & 0 \end{bmatrix}, \quad
    Z \coloneqq \begin{bmatrix} 1 & 0 \\ 0 & -1 \end{bmatrix}, \quad
    I \coloneqq \begin{bmatrix} 1 & 0 \\ 0 & 1 \end{bmatrix}.
\end{equation*}

\subsection{Uncloneable Encryption} \label{sec:uncloneable_encryption}

\emph{Uncloneable encryption}, as introduced in~\cite{BL20}, is a symmetric-key encryption scheme with \emph{quantum} ciphertext to prevent unauthorized replication.

\begin{definition}[Quantum Encryption of Classical Messages]\label{def:uncloneable_encryption}
    A \emph{quantum encryption of classical messages} (QECM) is a triple of efficient quantum algorithms $(\Gen, \Enc, \Dec)$ such that:
    \begin{itemize}
        \item $\Gen(1^\lambda) \mapsto k$: takes a security parameter $1^\lambda$ as input, and outputs a classical secret key $k \in [K] \subseteq \set{0,1}^\lambda$.\footnote{The notation ``$1^\lambda = 1 1 \cdots 1$'' means that the generation of the key $k \in [K]$ is efficient with respect to the security parameter~$\lambda$.}
        \item $\Enc(m,k) \mapsto \rho$: takes a message $m\in\set{0,1}$ and a secret key $k \in [K] \subseteq \set{0,1}^\lambda$ as inputs, and outputs a quantum ciphertext $\rho \in \B(\Hilbert)$.
        \item $\Dec(\rho,k) \mapsto m$: takes a ciphertext $\rho \in \B(\Hilbert)$ and a secret key $k \in [K] \subseteq  \set{0,1}^\lambda$ as inputs, and outputs a message $m\in\{0,1\}$.
    \end{itemize}
\end{definition}

The following definition requires that, for all keys that can be generated by the scheme, a valid ciphertext is always correctly decrypted, up to negligible probability.

\begin{definition}[Correctness] \label{def:uncloneable_encryption_correctness}
     A scheme $(\Gen, \Enc, \Dec)$ is said to be \emph{correct} if for any security parameter $\lambda$, for any secret key $k = \Gen(1^\lambda)$, and for any message $m$, we have
    \begin{equation*}
        \Pr\ofAlt[\Big]{ \Dec\of[\big]{ \Enc(m, k), k } = m } = 1 - \negl(\lambda)\,,
    \end{equation*}
    where $\negl(\lambda)$ is a \emph{negligible function} of $\lambda$.
\end{definition}

Now that we have a formal definition of the scheme  and its correctness, we can define the security property called \emph{uncloneable-indistinguishability}. This security property is formalized in terms of a no-cloning game for a $1$-bit message (see~\Cref{fig:uncloneable_encryption}).

\begin{definition}[No-Cloning Game] \label{def:cloning_game}
    Given an uncloneable encryption scheme $(\Gen, \Enc, \Dec)$ and a security parameter $\lambda$, the \emph{no-cloning game} involving a challenger $\Alice$ and three adversaries $(\Pirate, \Bob, \Charlie)$ is defined by the following procedure:
    \begin{enumerate}
        \item A challenger $\Alice$ generates a key $k \leftarrow \Gen(1^\lambda)$ and a message $m \in \set{0, 1}$ uniformly at random, and sends the quantum state $\rho \leftarrow \Enc(m,k)$ to $\Pirate$.
        \item The adversary $\Pirate$ applies a \CPTP~map $\Phi: \B(\Hilbert_\Alice) \to \B(\Hilbert_{\Bob} \otimes \Hilbert_{\Charlie})$ on the state $\rho$ to obtain the bipartite state $\Phi(\rho)$, and sends to $\Bob$ and $\Charlie$ their respective register.
        \item The adversaries $\Bob$ and $\Charlie$ receive the secret key $k$ and measure their state using two {\POVM}s $\set{B_k}$ and $\set{C_k}$, to output $m_\Bob, m_\Charlie \in \set{0,1}$.
        \item The adversaries $(\Pirate, \Bob, \Charlie)$ win if $m = m_\Bob = m_\Charlie$.
    \end{enumerate}
\end{definition}

\begin{remark}
 Since this work concerns the question of the \emph{uncloneable bit}, we consider in \Cref{def:cloning_game} a game with single-bit messages $m \in \{0,1\}$. See \cite{BL20} for the more general case $m \in \{0,1\}^n$.
\end{remark}

The uncloneable-indistinguishability security was introduced by Broadbent and Lord in~\cite{BL20}, combining two other notions of security: uncloneable security, which they also introduced, and indistinguishable security~\cite{GM84,ABF+16}. In contrast to~\cite{BL20}, we propose a stronger security definition, wherein adversaries are unbounded.

\begin{definition}[Uncloneable-Indistinguishable Security] \label{def:uncloneable_encryption_security}
    An uncloneable encryption scheme 
    satisfies \emph{strong uncloneable-indistinguishable security} if for any security parameter $\lambda$ and for any adversaries $(\Pirate, \Bob, \Charlie)$ the no-cloning game from \Cref{def:cloning_game} cannot succeed with probability more than $\sfrac{1}{2} + \negl(\lambda)$. If instead the adversaries cannot succeed with probability more than $\sfrac{1}{2} + f(\lambda)$ for some function $f: \RR \to \RR$ such that $\lim_\lambda f(\lambda) = 0$, we simply say that the scheme satisfies \emph{weak uncloneable-indistinguishable security}.
\end{definition}

\begin{remark}
    As sketched in~\cite{BL20}, uncloneable-indistinguishable security implies indistinguishable security in terms of a conventional encryption scheme~\cite{GM84,ABF+16}. For completeness, we provide a full proof of the indistinguishable security of our candidate scheme in~\Cref{sec:indistinguishability}.
\end{remark}

\subsection{Bound for Unconditional No-Cloning Game}

By definition of the no-cloning game (see  \Cref{def:cloning_game}), the winning probability of the three adversaries~$(\Pirate, \Bob, \Charlie)$ is equal to
\begin{align*}
    \Pr[\text{win}] &= \sup_{\substack{\Phi \\ B_{i|k}, C_{j|k}}} \: \E_{\substack{m \in \set{0,1} \\ k \leftarrow \Gen(1^\lambda)}} \: \sum_{i,j \in \set{0,1}} \mathbbm{1}_{\set{i = j = m}}\Tr\ofAlt[\big]{ \Phi(\rho_{m, k}) (B_{i|k} \otimes C_{j|k}) } \\
    &= \sup_{\Phi, B, C} \E_{m, k} \Tr\ofAlt[\big]{ \Phi(\rho_{m, k}) (B_{m|k} \otimes C_{m|k}) }\,,
\end{align*}
with $\rho_{m, k} \coloneqq \Enc(m, k)$, and where the expected values are taken with respect to the uniform measures, and the supremum is taken over all \CPTP~$\Phi: \B(\Hilbert_\Alice) \to \B(\Hilbert_{\Bob} \otimes \Hilbert_{\Charlie})$ (for all finite-dimensional Hilbert spaces $\Hilbert_\Bob$ and $\Hilbert_\Charlie$), and all families of \POVM~$\set{B_{i|k}}$ and $\set{C_{j|k}}$.

Given a \CPTP~map $\Phi: \B(\CC^d) \to \B(\CC^{d'})$, the Choi matrix~\cite{Wat18} $C_\Phi \in \M_d \otimes \M_{d'}$ is defined by
\begin{equation*}
    C_\Phi \coloneqq \sum_{i,j \in [d]} \ketbra{i}{j} \otimes \Phi\of[\big]{ \ketbra{i}{j} }\,, 
\end{equation*}
and satisfies the partial trace equality
\begin{equation} \label{eq:Choi_matrix_equation}
    \Phi(X) = \Tr_d\ofAlt[\big]{ C_\Phi (X^\T \otimes I_{d'}) }\,, \qquad \forall X \in \M_d\,,
\end{equation}
and Choi's theorem~\cite{Cho75}:
\begin{equation} \label{thm:Choi}
    \Phi \in \CPTP \qquad \Longleftrightarrow \qquad C_\Phi \succeq 0 \quad \text{ and }\quad \Tr_{d'}[C_\Phi] = I_d.
\end{equation}

Using the Choi matrix of the \CPTP~$\Phi: \B(\Hilbert_\Alice) \to \B(\Hilbert_{\Bob} \otimes \Hilbert_{\Charlie})$, with $\Hilbert_\Alice \simeq \CC^d$ and $\Hilbert_\Bob \otimes \Hilbert_\Charlie \simeq \CC^{d'}$, and the partial trace equality~\cref{eq:Choi_matrix_equation}, the winning probability of the no-cloning game becomes
\begin{equation*}
    \Pr[\text{win}] = \sup_{C_\Phi, B, C} \E_{m, k} \Tr\ofAlt[\big]{ C_\Phi (\rho^\T_{m, k} \otimes B_{m|k} \otimes C_{m|k}) }\,,
\end{equation*}
where the supremum is now taken over all $C_\Phi \succeq 0$ such that $\Tr_{d'}[C_\Phi] = I_d$. This last condition can be relaxed to $\Tr[C_\Phi] = d$, giving the first upper-bound on the winning probability:
\begin{align}
    \Pr[\text{win}] &\leq \sup_{C_\Phi, B, C} \E_{m, k} \Tr\ofAlt[\big]{ \tfrac{1}{d} \cdot C_\Phi (d \cdot \rho^\T_{m, k} \otimes B_{m|k} \otimes C_{m|k}) } \nonumber \\
    &= \sup_{\sigma, B, C} \E_{m, k} \Tr\ofAlt[\big]{ \sigma\, (d \cdot \rho^\T_{m, k} \otimes B_{m|k} \otimes C_{m|k}) }\,, \label{eq:winning_probability_upperbound_1}
\end{align}
where the first supremum is taken over all $C_\Phi \succeq 0$, and the last supremum is taken over all $\sigma \succeq 0$ such that $\Tr[\sigma] = 1$, \ie a quantum state. 
The upper-bound~\cref{eq:winning_probability_upperbound_1} is saturated on pure quantum states $\ketbra{\psi}{\psi}$, as pure quantum states form the extreme points of the convex set of quantum states, and the optimization is linear in $\sigma$. Hence
\begin{align}
    \Pr[\text{win}] &\leq \sup_{\psi, B, C} \bra{\psi} \E_{m, k}\ofAlt[\big]{ d \cdot \rho^\T_{m, k} \otimes B_{m|k} \otimes C_{m|k} } \ket{\psi} \label{eq:winning_probability_upperbound_2} \\
    &\leq \sup_{\psi, B, C} \big| \bra{\psi} \E_{m, k}\ofAlt[\big]{ d \cdot \rho^\T_{m, k} \otimes B_{m|k} \otimes C_{m|k} } \ket{\psi} \big| \label{eq:winning_probability_upperbound_3}  \\
    &= \sup_{B, C} \operatornorm[\Big]{ \E_{m, k}\ofAlt[\big]{ d \cdot \rho^\T_{m, k} \otimes B_{m|k} \otimes C_{m|k} } }, \label{eq:winning_probability_upperbound_4}
\end{align}
where the first two suprema are taken over all $\norm{\psi} = 1$ and the last equality holds because $d \cdot \rho^\T_{m, k} \otimes B_{m|k} \otimes C_{m|k}$ is a Hermitian operator.

By Naimark’s theorem~\cite{Per90b}, we may always assume that the {\POVM}s are in fact {\PVM}s up to increasing the dimension. 
Moreover, any adversaries $(\Pirate, \Bob, \Charlie)$ can be assumed to be \emph{symmetric}, \ie $\Hilbert_\Bob = \Hilbert_\Charlie$ and $\set{B_{i|k}} = \set{C_{j|k}}$, by taking $\Bob$'s and $\Charlie$’s spaces to be $\Hilbert_\Bob \oplus \Hilbert_\Charlie$, the \PVM~operators to be $\set{B_{i|k} \oplus C_{i|k}}$,\footnote{We do not consider the cases where $i\neq j$ because then the winning probability of the adversaries is zero.} and the \CPTP~map $\Phi$ that sends $\Bob$’s part of the output to the first component of the direct sum space and $\Charlie$’s to the second component. Hence, the upper-bound~\cref{eq:winning_probability_upperbound_4} becomes
\begin{equation} \label{eq:winning_probability_upperbound_5}
    \Pr[\text{win}] \leq \sup_{M} \operatornorm[\Big]{ \E_{m, k}\ofAlt[\big]{ d \cdot \rho^\T_{m, k} \otimes M_{m|k} \otimes M_{m|k} } },
\end{equation}
where the supremum is taken over all {\PVM}s $\set{M_{i|k}}$.

\subsection{Link with Monogamy-of-Entanglement Games}

Limitations imposed by the \emph{no-cloning theorem}~\cite{Wer98,KW99} are closely related to limitations due to \emph{monogamy-of-entanglement} (MoE)~\cite{CKW00,KW04}. This explains the close relationship between no-cloning games and monogamy-of-entanglement games (see~\Cref{fig:MoE_Game}), introduced in the foundational paper~\cite{TFKW13}.

\begin{figure}[t]
    \centering
    \begin{tikzpicture}
        \newcommand{\fillproportion}{20}
        \tikzset{
            block/.style={draw, minimum width=3em, minimum height=3em, font=\sffamily\Large, circle, line width=1mm},
            arrow/.style={-stealth, ultra thick},
            double_arrow/.style={-{Implies[]}, very thick, double},
        }
    
        \draw[rounded corners=15pt, thick, draw=Gray, fill=Gray!\fillproportion, line width=0.5mm] ([xshift=-2.4em, yshift=-7em] A) rectangle node[font=\Large, xshift=0em, yshift=-6.5em] {$\sigma_{\scriptscriptstyle \Alice \Bob \Charlie}$} ([xshift=-5.7em, yshift=3.5em] B);
        
        \node[block, draw=DarkRed, fill=DarkRed!\fillproportion] (A) {A};
        \node[block, draw=DarkGreen, fill=DarkGreen!\fillproportion, above right=1.5em and 6em of A] (B) {B};
        \node[block, draw=DarkBlue, fill=DarkBlue!\fillproportion, below right=1.5em and 6em of A] (C) {C};

        \draw[dashed, very thick] ([xshift=5em] A.center) -- ([xshift=14em] A.center);
        
        \draw[double_arrow] (A.30) -- node[above, yshift=2pt] {$k$} ([xshift=-2pt] B.200);
        \draw[double_arrow] (A.330) -- node[below, yshift=-2pt] {$k$} ([xshift=-2pt] C.160);
    
        \draw[double_arrow] (A.west) -- ([xshift=-2.5em] A.west) node [label = left:{$m_{\Alice}$}] {};
        
        \draw[double_arrow] (B.east) -- ([xshift=2.5em] B.east) node [label = right:{$m_{\Bob}$}] {};
        
        \draw[double_arrow] (C.east) -- ([xshift=2.5em] C.east) node [label = right:{$m_{\Charlie}$}] {};
    \end{tikzpicture}
    \caption{\emph{Monogamy-of-Entanglement Game.} Alice~(\Alice), Bob~(\Bob), and Charlie~(\Charlie) share a quantum state~$\sigma_{\Alice\Bob\Charlie}$. 
    Given a classical random key~$k\in\set{1,..,K}$, 
    Alice performs a measurement $\set{A_{m|k}}_m$ and obtains~$m_\Alice\in\set{0,1}$. 
    Using the same key $k$, the players Bob and Charlie perform respective measurements $\set{B_{m|k}}_m$ and $\set{C_{m|k}}_m$, 
    and we say that they win the game if both of them recover the exact same message as Alice, \ie if $m_\Alice\!=\!m_\Bob\!=\!m_\Charlie$.}
    \label{fig:MoE_Game}
\end{figure}
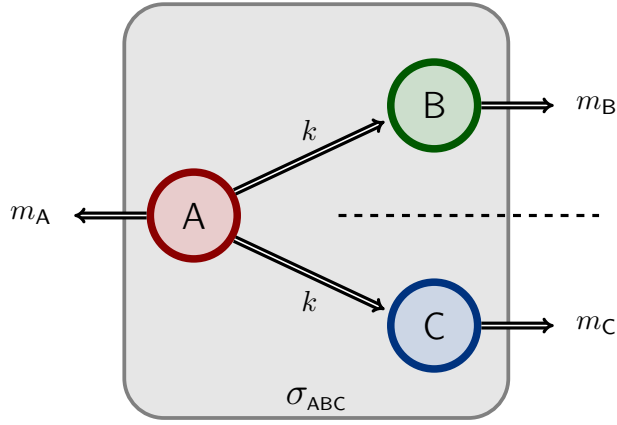

An MoE game involves three parties: a referee, denoted as Alice ($\Alice$), and two cooperating players, Bob ($\Bob$) and Charlie ($\Charlie$). The game proceeds as follows. Alice initially has a fixed family of positive operator-valued measures. Before the game's start, Bob and Charlie jointly agree upon a strategy, preparing a tripartite quantum state and sending the corresponding marginal state to referee Alice.
Once the game begins, the players are space-like separated, meaning that any form of communication is no longer allowed. The referee Alice picks a \POVM uniformly at random and performs a measurement on her quantum state to produce some classical outcome $m_\Alice$. She publicly announces the \POVM she used to Bob and Charlie. Their goal is then to independently recover $m_\Alice$, by performing some measurements on their respective states, which result in classical outcomes $m_\Bob$ and $m_\Charlie$, respectively.
Finally, the referee $\Alice$ declares that the players Bob and Charlie won the game if both of them recover Alice's outcome, \ie if exactly $m_\Alice=m_\Bob=m_\Charlie$.

When the \POVM{}s fixed by the referee Alice are the ones of conjugate coding~\cite{Wie83}, it is easy to show that the winning probability for the MoE game is at least $\frac{1}{2} + \frac{1}{2\sqrt2}$: the players Bob and Charlie share no entanglement, send to Alice the state $\cos \frac{\pi}{8} \ket{0} + \sin \frac{\pi}{8} \ket{1}$ from the Breidbart basis, and always output $m_\Bob = m_\Charlie = 0$, which is indeed equal to $m_\Alice$ with probability $(\cos\frac\pi8)^2 = \frac12 + \frac{1}{2\sqrt2}$.\footnote{Alternatively, this can be seen as a consequence of an entropic uncertainty relation due to Deutsch~\cite{Deu83}.}
Interestingly, this value is actually the optimal one, and if we repeat $n$ times this game in parallel, then the optimal winning probability is precisely $\of[\big]{ \frac{1}{2} + \frac{1}{2\sqrt2} }^n$~\cite{TFKW13}.

In the absence of MoE, the players in an  MoE game would be able to perfectly win the game by considering a state maximally entangled between the referee $(\Alice)$ and the two players $(\Bob)$ and $(\Charlie)$, \ie a tripartite state $\rho_{\Alice\Bob\Charlie}$ with reduced states
\begin{equation*}
    \rho_{\Alice\Bob} = \frac{1}{2} \sum_{ij} \ketbra{ii}{jj} \qquad \text{ and } \qquad \rho_{\Alice\Charlie} = \frac{1}{2} \sum_{ij} \ketbra{ii}{jj}\,.
\end{equation*}
The existence of such a state is precluded by the MoE principle, which stipulates that no quantum state can be maximally entangled with more than two parties simultaneously. This principle also implies the no-cloning principle which states that it is not possible, in general, to copy a quantum state. If this were possible, then a quantum channel $\Phi: \B(\Hilbert) \to \B(\Hilbert \otimes \Hilbert)$, with marginals $\Phi_1 (\rho) = \rho$ and $\Phi_2 (\rho) = \rho$ for all $\rho \in \B(\Hilbert)$, would exist. But then its Choi matrix $C_{\Phi}$ is such that the marginals become
\begin{equation*}
    C_{\Phi_1} = \sum_{ij} \ketbra{ii}{jj} \qquad \text{ and } \qquad C_{\Phi_2} = \sum_{ij} \ketbra{ii}{jj},
\end{equation*}
that is, both proportional to a maximally entangled state, which is not possible due to the MoE and Choi's theorem~\cref{thm:Choi}.

This relationship between MoE and no-cloning is reflected in the winning probability of MoE games and no-cloning games. The winning probability of the MoE game can be described as follows:
\begin{equation} \label{eq:winning_probability_MoE_game}
    \Pr\big[\text{win}\big]
    \,=\,
    \sup_{\sigma, B, C}\, \frac{1}{K} \sum_{m, k} \Tr\ofAlt[\big]{ \sigma\, (A_{m|k} \otimes B_{m|k} \otimes C_{m|k}) }\,,
\end{equation}
where $\sigma$ is a quantum state in $\Hilbert_\Alice\otimes\Hilbert_\Bob\otimes\Hilbert_\Charlie$, and the sets $\set{A_{m|k}}_k$, $\set{B_{m|k}}_k$, $\set{C_{m|k}}_k$ are \POVM{}s on respectively $\Hilbert_\Alice$, $\Hilbert_\Bob$, $\Hilbert_\Charlie$. 
When $m\in\{0,1\}$, this winning probability is precisely the expression we have in~\cref{eq:winning_probability_upperbound_1} when $A_{m|k}=(d/2) \cdot \rho^\T_{m, k}$ is a \POVM, \ie when the  normalization condition $\sum_m \rho_{m, k} = 2\,I_d/d$ holds for all $k$.

\section{Candidate Scheme} \label{sec:scheme}

In this section, we propose a new encryption scheme based on the Clifford algebra. In the next section, we show that it allows us to recover previously known results for $K=2$, and we provide strong numerical evidence that this scheme achieves the uncloneable security for large~$K$.

\subsection{Clifford Algebra}
\label{sec:Clifford_Algebra}

For our uncloneable encryption scheme, we  make use of the structure of Clifford algebra~\cite{Lou01, DL03}, widely used in quantum information theory, in particular as a generalization of the Bloch sphere representation~\cite{Die06, WW08b}, in operator algebras theory~\cite{Pis03, BN18}, and in non-local games~\cite{Slo11,Ost16}.
This choice is motivated by the fact that the Clifford generators used below, were shown in \cite{BN18} to give rise to maximally incompatible binary observables. Thus, the measurements associated with different keys are extremal from the viewpoint of joint measurability, making them natural candidates for an uncloneable encryption scheme, where security should arise from the impossibility of simultaneously reproducing such incompatible measurement outcomes.

For any integer $n \in \NN_+$, the free real associative algebra generated by $\Gamma_1, \ldots, \Gamma_{n}$, and subject to the anti-commutation relations:
\begin{equation*}
    \anticommutator{\Gamma_i}{\Gamma_j} \coloneqq \Gamma_i \Gamma_j + \Gamma_j \Gamma_i = 2\, \delta_{ij} \mathbbm{1}\,,
\end{equation*}
is called the \emph{Clifford algebra}, and denoted $\mathbf{CL}_n$. Any irreducible representation of the Clifford algebras $\mathbf{CL}_{2n}$ and $\mathbf{CL}_{2n+1}$ can be constructed explicitly by the following $2n + 1$ Pauli string operators acting on~$\CC^{2^n}$, called the Jordan-Wigner transformation~\cite{JW93}:
\begin{align*}
    \sigma_{n,2i-1} &= X^{\otimes (i-1)} \otimes Y \otimes I^{\otimes (n-i)}, \quad i \in \set{1, \ldots, n}\,, \\
    \sigma_{n,2i} &= X^{\otimes (i-1)} \otimes Z \otimes I^{\otimes (n-i)}, \quad i \in \set{1, \ldots, n}\,, \\
    \sigma_{n,2n+1} &= X^{\otimes n}\,,
\end{align*}
These operators are traceless, and as Hermitian unitaries, they have spectrum $\pm 1$. Mapping the generators $\Gamma_1, \ldots, \Gamma_{2n}$ to $\sigma_{n,1}, \ldots, \sigma_{n,2n}$  gives the irreducible representation of the even Clifford algebra $\mathbf{CL}_{2n}$. Mapping the generators $\Gamma_1, \ldots, \Gamma_{2n+1}$ to $\sigma_{n,1}, \ldots, \sigma_{n,2n+1}$ or $\sigma_{n,1}, \ldots, \shortminus \sigma_{n,2n+1}$ gives the two inequivalent irreducible representations of the odd Clifford algebra $\mathbf{CL}_{2n+1}$.

An important property of the Clifford algebra is that $2n$ Hermitian anti-commuting operators $\Gamma_1, \ldots, \Gamma_{2n}$ can be view as orthogonal vectors, for the normalized Frobenius inner product, forming a basis for the $2n$-dimensional real vector space they span. Indeed, given $u, v \in \RR^{2n}$ we have
\begin{equation*}
    \bigg\langle \sum^{2n}_{i=1} u_i \Gamma_i, \sum^{2n}_{j=1} v_j \Gamma_j \bigg\rangle = \frac{1}{2n} \of[\bigg]{ \sum^{2n}_{i=1} u_i v_i \underbracket[0.75pt]{\,\Tr\of[\big]{\Gamma^2_i}\,}_{=2n} + \sum^{2n}_{\substack{i,j=1 \\ i < j}} (u_i v_j - u_j v_i) \underbracket[0.75pt]{\,\Tr\of[\big]{\Gamma_i \Gamma_j}\,}_{=0} } = \langle u, v \rangle\,.
\end{equation*}
In particular for a unit vector $v \in \RR^{2n}$, the operator $\sum^{2n}_{i=1} v_i \Gamma_i$ is a Hermitian unitary. Using this property, given any vector $v \in \RR^{2n}$, we have
\begin{equation}   \label{eq:formula-with-norm-2}
    \operatornorm[\bigg]{\sum^{2n}_{i=1} v_i \Gamma_i} = \norm{v}_2\,,
\end{equation}
and taking the all-ones vector, we get $\operatornorm[\big]{\sum^{2n}_{i=1} \Gamma_i} = \sqrt{2n}$.

\subsection{Definition of the Scheme}
        
We  now define our candidate scheme for a weak uncloneable bit, using the Clifford algebra introduced in \Cref{sec:Clifford_Algebra}.

\begin{definition}[Candidate Uncloneable Bit Encryption] \label{def:Clifford_uncloneable_encryption}
    Let $\lambda \in \NN_+$. We define our candidate uncloneable bit scheme. $\Gen(1^\lambda)$ samples uniformly a key  $k \in [K]$  with $K = 2\lambda$ ($\lambda$ even) or $K = 2\lambda +1 $ ($\lambda$ odd). The algorithms $\Enc$ and $\Dec$ are defined in \Cref{alg:enc} and \Cref{alg:dec}, respectively. They use the irreducible representation of the generators $\Gamma_1, \ldots, \Gamma_K$ of a Clifford algebra (\Cref{sec:Clifford_Algebra})\footnote{When $\lambda$ is odd, we choose $\sigma_{\lambda, 1}, \ldots, \sigma_{\lambda, 2\lambda+1}$, even if both irreducible representations form a valid scheme.}, acting on $\CC^{d}$ where $d=2^\lambda$.

    \noindent
    \begin{minipage}[t]{0.5\linewidth}
        \null 
        \begin{algorithm}[H]
            \DontPrintSemicolon
            \caption{Encryption $\Enc(m, k)$.} \label{alg:enc}
        	\SetKwInOut{Input}{Input}
        	\SetKwInOut{Output}{Output}
        
        	\vbox to 8.5em {
        	   \Input{A message $m \in \set{0, 1}$ and a key $k \in [K]$.}
        	   \Output{A state $\rho_{m,k}$ acting on $\CC^d$.}
        
                Compute $\rho_{m,k} = \frac{2}{d}\, \frac{I_d + (\shortminus 1)^m \Gamma_k}{2}$, the normalized projector onto the $(\shortminus 1)^m$-eigenspace of $\Gamma_k$.\;
        	   Output $\rho_{m,k}$.\;
            }
        \end{algorithm}
    \end{minipage}% DO NOT REMOVE THIS COMMENT
    \begin{minipage}[t]{0.5\linewidth}
        \null
        \begin{algorithm}[H]
            \DontPrintSemicolon
            \caption{Decryption $\Dec(\rho, k)$.}\label{alg:dec}
        	\SetKwInOut{Input}{Input}
        	\SetKwInOut{Output}{Output}
        
        	\vbox to 8.5em {
        	   \Input{A state $\rho$ acting on $\CC^d$ and a key $k \in [K]$.}
        	   \Output{A message $m \in \set{0, 1}$.}
        
                Measure $\rho$ in the eigenbasis of $\Gamma_k$, with the \PVM $\set{\frac{1}{2}\of{I_d + (\shortminus 1)^i \Gamma_k}}_i$. Call $m$ the outcome.\;
        	   Output $m$.\;
            }
        \end{algorithm}
    \end{minipage}
\end{definition}

\begin{remark}
    The correctness~\Cref{def:uncloneable_encryption_correctness} is immediate since the operators $\frac{I_d + \Gamma_k}{2}$ and $\frac{I_d - \Gamma_k}{2}$ are orthogonal. Thus the measurement of $\rho_{m,k}$ is
    \begin{equation*}
        \Tr \bigg[ 
            \of[\bigg]{\frac{I_d + (\shortminus 1)^i \Gamma_k}{2}} \; 
            \rho_{m,k} \; 
            \of[\bigg]{\frac{I_d + (\shortminus 1)^i \Gamma_k}{2}} 
        \bigg] =
        \begin{cases}
            1 &\text{if } i = m\,, \\
            0 &\text{otherwise}\,,
        \end{cases}
    \end{equation*}
    and $\Pr\ofAlt[\Big]{ \Dec\of[\big]{ \Enc(m, k), k } = m } = 1$.
\end{remark}

When $K=2$, with $\Gamma_1 \coloneqq X$ and $\Gamma_2 \coloneqq Z$, then
\begin{equation*}
    \frac{I_d + (\shortminus 1)^m \Gamma_1}{2} = H \ketbra{m}{m} H^* \qquad \text{ and } \qquad \frac{I_d + (\shortminus 1)^m \Gamma_2}{2} = \ketbra{m}{m}\,,
\end{equation*}
which is exactly the encryption used by the scheme defined in~\cite{BL20}. It should also be noted that in this case, the state $\rho_{m,k}$ is pure; however, it becomes mixed for larger $K$, as required by the impossibility results~\cite{MST21arxiv, AKL+22}.

\subsection{Conjecture}

The upper-bound on the winning probability of the no-cloning game for three adversaries $(\Pirate, \Bob, \Charlie)$ is given by~\cref{eq:winning_probability_upperbound_5}:
\begin{equation*}
    \Pr[\text{win}] \leq \sup_{M} \operatornorm[\Big]{ \E_{m, k}\ofAlt[\big]{ d \cdot \rho^\T_{m, k} \otimes M_{m|k} \otimes M_{m|k} } }\,,
\end{equation*}
where the supremum is taken over all \PVM~$\set{M_{i|k}}$. For all the keys $k$, the \PVM~$\set{M_{i|k}}$ are binary measurement operators of dimension $D$. We can write
\begin{equation*}
    M_{i|k} = \frac{I_D + (\shortminus 1)^i U_k}{2}\,,
\end{equation*}
for some Hermitian unitaries (observables) $U_k \coloneqq M_{0|k} - M_{1|k}$. Then the upper-bound~\cref{eq:winning_probability_upperbound_5} for Clifford uncloneable encryption~(\Cref{def:Clifford_uncloneable_encryption}) is
\begin{align}
    \Pr[\text{win}] &\leq \sup_{\{U_k\}} \operatornorm[\Bigg]{ \E_{m, k}\ofAlt[\bigg]{ d \cdot \frac{2}{d} \frac{I_d + (\shortminus 1)^m \Gamma_k}{2} \otimes \frac{I_D + (\shortminus 1)^m U_k}{2} \otimes \frac{I_D + (\shortminus 1)^m U_k}{2} } } \nonumber \\
    &\leq \sup_{\{U_k\}} \frac{1}{2 K} \operatornorm[\bigg]{ \sum_{m, k} d \cdot \frac{2}{d} \frac{I_d + (\shortminus 1)^m \Gamma_k}{2} \otimes \frac{I_D + (\shortminus 1)^m U_k}{2} \otimes \frac{I_D + (\shortminus 1)^m U_k}{2} } \nonumber  \\
    &\leq \sup_{\{U_k\}} \frac{1}{2 K} \bigg\lVert \frac{1}{4} \sum_{m, k} 
    \begin{aligned}[t]
      & \Big(I_d \otimes I_D \otimes I_D \\
      +& (\shortminus 1)^m \, \of[\big]{\Gamma_k \otimes I_D \otimes I_D + I_d \otimes U_k \otimes I_D + I_d \otimes I_D \otimes U_k} \\
      +& (\shortminus 1)^{2m} \of[\big]{\Gamma_k \otimes U_k \otimes I_D + \Gamma_k \otimes I_D \otimes U_k + I_d \otimes U_k \otimes U_k} \\
      +& (\shortminus 1)^{3m} \, \Gamma_k \otimes U_k \otimes U_k \Big){\bigg\lVert}_{\text{op}}
    \end{aligned} \nonumber  \\
    &\leq \frac{1}{4} + \frac{1}{4 K} \sup_{\{U_k\}} \operatornorm[\bigg]{\sum_k \Big(\Gamma_k \otimes (U_k \otimes I_D + I_D \otimes U_k) + I_d \otimes U_k \otimes U_k \Big)}\,,\label{eq:clifford_winning_probability_upperbound}
\end{align}
where the supremum is taken over all Hermitian unitaries $U_k \in \U(\CC^D)$.

A naive triangular inequality of the operator norm yields the trivial upper-bound $\Pr[\text{win}]\leq \sfrac{1}{4} + \sfrac{3K}{4 K}=1$. Therefore, we need a finer upper-bound of this quantity, which we conjecture to be of the following form:

\begin{conjecture} \label{conj:main_conjecture}
    Let $\Gamma_1, \ldots, \Gamma_K$ be any pairwise anti-commuting Hermitian unitaries (Hermitian representation of some Clifford algebra) of dimension $d$. Then for all Hermitian unitaries $\set{U_k}$ of dimension $D$, the following upper-bound holds:
    \begin{equation}  \label{eq:conjecture}
        \operatornorm[\bigg]{\sum_{k \in [K]} \Big(\Gamma_k \otimes (U_k \otimes I_D + I_D \otimes U_k) + I_d \otimes U_k \otimes U_k\Big)} \leq K + 2 \sqrt{K}\,.
    \end{equation}
\end{conjecture}

\begin{remark*}
    Under this conjecture, the upper-bound in~\cref{eq:clifford_winning_probability_upperbound} gives $\frac{1}{2} + \frac{1}{2\sqrt{K}}$, with the asymptotic limit $\lim_{K \to \infty} \Pr[\text{win}] = \sfrac{1}{2}$ at a rate $O(\sfrac{1}{\sqrt{K}})$, which would prove that the Clifford encryption scheme is weakly secure 
    (see \Cref{def:uncloneable_encryption_security}).
\end{remark*}

\subsection{Basic Properties}

In this subsection, we prove basic properties related to our conjecture. First, we prove that the identity matrix saturates the bound of the conjecture, which means that if the wanted upper-bound holds, then it is tight (\Cref{{prop:sup-achieved-at-Identity}}). Then we prove that the conjecture  is true if the adversary's strategy is reduced to commuting operators $U_k$ (\Cref{prop:the-conjecture-is-true-for-commuting-operators}) or if they only use product states (\Cref{prop:Conjecture-true-for-product-states}).
        
In what follows, we will refer to the family of operators involved in~\Cref{conj:main_conjecture} as:
\begin{equation}  \label{eq:Definition-of-W}
    W_K(U_1, \ldots, U_K) \coloneqq \sum_{k \in [K]} \Big(\Gamma_k \otimes (U_k \otimes I_D + I_D \otimes U_k) + I_d \otimes U_k \otimes U_k\Big)\,.
\end{equation}
We simply write $W \coloneqq W_K(U_1, \ldots, U_K)$ if there are no ambiguities on $K \in \NN_+$ and the Hermitian unitaries $U_1, \ldots, U_K$.

\begin{proposition}[Tight Upper-Bound]  \label{prop:sup-achieved-at-Identity}
    The supremum of $\operatornorm{W_K}$ is lower-bounded by $K+2\sqrt{K}$, achieved by taking $U_k=I_D$ for all $k\in[K]$.
    As a consequence, the upper-bound in \Cref{conj:main_conjecture} can only be tight.
\end{proposition}

\begin{proof}
    Combine the formulae $\operatornorm{\sum_i \Gamma_i+I}=\operatornorm{\sum_i \Gamma_i}+1$, valid since the spectrum of $\sum_i \Gamma_i$ is symmetric~\cite{HKMS19}, and \cref{eq:formula-with-norm-2} to obtain the value $K+2\sqrt{K}$.
\end{proof}

\begin{remark}[Low Rank Operators Also Saturate the Bound]
    Actually, the value $K+2\,\sqrt{K}$ is also achieved by taking any $U_k\in\CC^D$ with projector onto the $1$-eigenspace of rank $r\leq\frac{D-1}{K}$.
    Indeed, write
    \begin{align*}
        U_k = 2 \Bigg(\sum_{i=1}^{r} |u^{(i)}_k\rangle\langle u^{(i)}_k|\Bigg) - I_{D}\,,
    \end{align*}
    for all $k\in[K]$ and for some quantum states $|u^{(i)}_k\rangle\in\CC^{D}$ (or the zero vector).
    Using the condition $D\geq K\,r+1$, there exists a quantum state $|u^\perp\rangle\in\CC^{D}$ that is orthogonal to all the $|u^{(i)}_k\rangle$ for $i\in[r]$ and $k\in[K]$. Note that this vector satisfies 
    $\langle u^\perp|U_k|u^\perp\rangle = 2 (\sum_i 0) - 1 = -1$ for all $k$. 
    It follows that
    \begin{align*}
        &\operatornorm[\big]{W_K}
        \,=\,
        \sup_{\substack{|\psi\rangle\in\CC^{d}\otimes\CC^{D}\otimes\CC^{D} \\ \norm{|\psi\rangle}=1}} \Big|\big\langle\psi\,\big|\, \textstyle\sum_{k} \Big(\Gamma_k \otimes U_k \otimes I_{D} + \Gamma_k\otimes I_{D} \otimes U_k + I_{d} \otimes U_k \otimes U_k \Big)\,\big|\,\psi\big\rangle\Big|\\
        &\geq
        \sup_{\substack{|a\rangle\in\CC^{d} \\ \norm{|a\rangle}=1}} \,\,\Big|\big\langle a\otimes u^\perp\otimes u^\perp\,\big|\, \textstyle\sum_{k} \Big(\Gamma_k \otimes (U_k \otimes I_{D} + I_{D} \otimes U_k) + I_{d} \otimes U_k \otimes U_k \Big) \,\big|\, a\otimes u^\perp\otimes u^\perp\big\rangle\Big|\\
        &=\,\sup_{\substack{|a\rangle\in\CC^{d} \\ \norm{|a\rangle}=1}} \Big|2\,\langle a\,|\,\textstyle \sum_{k} -\Gamma_k\,|\,a\rangle+K\Big| \\
        &=\,2\,\operatornorm[\Big]{\textstyle \sum_{k} -\Gamma_k} \!\!+ K
        \,=\,
        2\,\norm[\Bigg]{\begin{bsmallmatrix} -1 \\ \vdots \\ -1 \end{bsmallmatrix}}_2 + K
        \,=\,
        2\sqrt{K}+K\,,
    \end{align*}
    using the symmetry of the spectrum of $\sum_k\Gamma_k$~\cite{HKMS19} in the third last equality and then 
    using \cref{eq:formula-with-norm-2} in the second last one. Hence the claimed result.
\end{remark}

\begin{proposition}[True for Commuting Operators]   \label{prop:the-conjecture-is-true-for-commuting-operators}
    If the operators $U_k$ commute, then \Cref{conj:main_conjecture} holds.
\end{proposition}

\begin{proof}
     If the operators $U_k$ commute, then they are diagonalizable in a common basis. But their eigenvalues are $\pm1$ because they are Hermitian and unitaries, so we may assume that they are of the form
     $$U_k \simeq \left( \begin{smallmatrix}
         \pm1 \\
         & \ddots\\
         && \pm1
     \end{smallmatrix} \right)\,.$$
     Then, using the triangular inequality, we obtain:
     \begin{align*}
        \operatornorm{W_K}
        \,&\leq\,
        \operatornorm[\Bigg]{\sum_{k=1}^K\Gamma_k\otimes (\pm1)\otimes 1}
        +\; \operatornorm[\Bigg]{\sum_{k=1}^K\Gamma_k\otimes 1 \otimes (\pm1)}
        +\; \sum_{k=1}^K\operatornorm[\Bigg]{\left(\begin{smallmatrix}
             \pm1 \\
             & \ddots\\
             && \pm1
         \end{smallmatrix}\right)} \\
        &=\, \sqrt{K}+\sqrt{K}+K\,,
     \end{align*}
     because $\operatornorm[\big]{\textstyle\sum_{k=1}^K\Gamma_k\otimes 1 \otimes (\pm1)}=\operatornorm[\big]{\textstyle\sum_{k=1}^K\Gamma_k}=\sqrt{K}$ using \cref{eq:formula-with-norm-2}.
\end{proof}

We give a connection between the operator norm of $W$ and the supremum of $\bra{\psi} W \ket{\psi}$:

\begin{lemma}[Expression of the Operator Norm] \label{prop:three_equalities}
    For all $K \in \NN_+$, we have the two following equalities:
    \begin{equation*}
        \sup_U \sup_{\norm{\psi}=1} \bra{\psi} W \ket{\psi} \quad \stackrel{(1)}{=} \quad \sup_U \sup_{\norm{\psi}=1} \lvert \bra{\psi} W \ket{\psi} \lvert \quad \stackrel{(2)}{=} \quad \sup_U \operatornorm{W}\,,
    \end{equation*}
    where the suprema are taken over all Hermitian unitaries $U_k \in \U(\CC^D)$.
\end{lemma}
\begin{proof}
    Let $W' \coloneqq W + K \cdot I_d \otimes I_D \otimes I_D$, then we can write,
    \begin{equation*}
        W' = \sum_k \Big(\Gamma_k \otimes I_D \otimes I_D + I_d \otimes U_k \otimes I_D\Big) \, \Big(\Gamma_k \otimes I_D \otimes I_D + I_d \otimes I_D \otimes U_k\Big)\,.
    \end{equation*}
    Each $\Gamma_k$ is a Hermitian unitary with a symmetric spectrum $\pm 1$. There exist an invertible matrix $H_k$ that can diagonalize $\Gamma_k$, \ie
    \begin{equation*}
        \Gamma_k = H_k
        \begin{pmatrix}
            - I_{d/2} & 0 \\
            0 & I_{d/2}
        \end{pmatrix}
        H_k^{\shortminus 1}\,.
    \end{equation*}
    Thus, under those changes of basis, the operator $W'$ is expressed as follows:
    \begin{equation*}
        W' = \sum_k \tilde{H}_k
        \begin{pmatrix}
            (I - U_k)^{\otimes2} & 0 \\
            0 & (I + U_k)^{\otimes2}
        \end{pmatrix}
        \tilde{H}_k^{\shortminus 1}\,,
    \end{equation*}
    with $\tilde{H}_k \coloneqq H_k \otimes I_D \otimes I_D$. Since for all $U_k$, the inequalities $-I_D \preceq U_k \preceq I_D$ hold, the two operators $I + U_k$ and $I - U_k$ are both positive semi-definite, so is each term of the sum, and thus $W' \succeq 0$.

    The first equality holds since,
    \begin{equation*}
        \sup_{\norm{\psi}=1} \lvert \bra{\psi} W \ket{\psi} \lvert = \max \Big\{ \sup\nolimits_{\norm{\psi}=1} \bra{\psi} W \ket{\psi}, - \inf\nolimits_{\norm{\psi}=1} \bra{\psi} W \ket{\psi} \Big\}\,,
    \end{equation*}
    but because $W' \succeq 0$, we have that $- \sup_U \inf_{\norm{\psi}=1} \bra{\psi} W \ket{\psi} = K$, and when all $U_k$ are equal to $I_D$, we already know that $\sup_{\norm{\psi}=1} \bra{\psi} W \ket{\psi} = K + 2 \sqrt{K}$ from~\Cref{prop:sup-achieved-at-Identity}.
    
    Finally, the second equality is always true for Hermitian operators \cite[Lemma 3.2.4]{Zim90}, and $W$ is a Hermitian operator, as a sum of tensor products of Hermitian operators $\Gamma_k$ and $U_k$.
\end{proof}

Hence, the three upper-bounds \cref{eq:winning_probability_upperbound_2,eq:winning_probability_upperbound_3,eq:winning_probability_upperbound_4} of the winning probability of the no-cloning game for three adversaries $(\Pirate, \Bob, \Charlie)$, are all equal, and  \Cref{conj:main_conjecture} can also be stated as the largest eigenvalue of the operator $W$.

From \Cref{prop:three_equalities}, see that an equivalent formulation of \Cref{conj:main_conjecture} is:
\[
    \sup_U \sup_{\norm{\psi}=1} \bra{\psi} W \ket{\psi}
    \,\leq\,
    K+2\,\sqrt{K}\,,
\]
where the quantum state $|\psi\rangle$ is taken in $\CC^d\otimes\CC^D\otimes\CC^D$. We show that the conjecture holds in the restricted case where $|\psi\rangle$ is a product state between Alice and $\set{\text{Bob, Charlie}}$:

\begin{proposition}[True for Product State]   \label{prop:Conjecture-true-for-product-states}
    If the state is of the form $|\psi\rangle=|\alpha_\Alice\rangle\otimes|\varphi_{\Bob\Charlie}\rangle$, then: 
    $$
    \sup_U \bra{\psi} W \ket{\psi} \leq K+2\,\sqrt{K}\,.
    $$
\end{proposition}

\begin{proof}
    We have:
    \begin{align*}
        &\sum_{k=1}^K\big\langle\alpha_\Alice\otimes\varphi_{\Bob\Charlie}\,\big|\,\Big(\Gamma_k \otimes (U_k \otimes I_D + I_D \otimes U_k) + I_d \otimes U_k \otimes U_k\Big)\,\big|\,\alpha_\Alice\otimes\varphi_{\Bob\Charlie}\big\rangle
        \\
        \,=\,&
        \sum_{k=1}^K \big\langle \alpha_\Alice \, \big| \, \Gamma_k \, \big| \, \alpha_\Alice \big\rangle \underbrace{\big\langle \varphi_{BC}\,|\, (U_k \otimes I_D + I_D \otimes U_k) \,|\,\varphi_{BC}\big\rangle}_{=:c_k} \\
        &\,+\, \sum_{k=1}^K \big\langle\alpha_\Alice\otimes\varphi_{\Bob\Charlie}\,\big|\,I_d \otimes U_k \otimes U_k \,\big|\,\alpha_\Alice\otimes\varphi_{\Bob\Charlie}\big\rangle \\
         \,=\,&
         \sum_{k=1}^K\big\langle\alpha_\Alice\,\big|\,c_k\,\Gamma_k\,\big|\,\alpha_\Alice \big\rangle  + \sum_{k=1}^K \underbrace{\big\langle\alpha_\Alice\otimes\varphi_{\Bob\Charlie}\,\big|\,I_d \otimes U_k \otimes U_k \,\big|\,\alpha_\Alice\otimes\varphi_{\Bob\Charlie}\big\rangle}_{\leq 1}
         \\
         \,\leq\,&
         \norm[\big]{(c_1, \dots, c_K)}_2 + K
         \\
          \,=\,&
          2\,\sqrt{K}+K\,,
    \end{align*}
    using \cref{eq:formula-with-norm-2} in the second last line.
\end{proof}

\begin{remark}
    When $|\psi\rangle$ is a product state in all its tensors, \ie when $|\psi\rangle = |\alpha_\Alice\otimes\beta_\Bob\otimes\gamma_\Charlie\rangle$, then this result is equivalent to the one in \Cref{prop:the-conjecture-is-true-for-commuting-operators}.
\end{remark}

\subsection{Indistinguishability} \label{sec:indistinguishability}

Analogously to the two variants of uncloneable-indistinguishability security presented in~\Cref{def:uncloneable_encryption_security}, namely those with a strong convergence rate and those with an arbitrary convergence rate, we can similarly define two corresponding notions of indistinguishability security.

The \emph{strong indistinguishability security} requires that any adversary, upon receiving the encryption of a message $m$ randomly chosen from a pair of messages, cannot predict the value of $m$ with a probability greater than negligibly close to $\sfrac{1}{2}$. If the adversary's probability of correctly predicting the encrypted message converges to $\sfrac{1}{2}$ at any arbitrary rate, we refer to this security notion as simply \emph{indistinguishability security}.

In this section, we  prove that our candidate  scheme~\Cref{def:Clifford_uncloneable_encryption} satisfies the latter indistinguishability security. Indeed, the success probability of such adversary is bounded by,
\begin{equation*}
    \Pr[\text{win}] \leq \sup_{\substack{\Phi \\ \set{M_0, M_1}}} \: \E_{\substack{m \in \set{0,1} \\ k \leftarrow \Gen(1^\lambda)}} \: \Tr\ofAlt[\big]{ \Phi(\rho_{m, k}) M_m }\,,
\end{equation*}
with $\rho_{m, k} \coloneqq \Enc(m, k)$, and where the expected values are taken with respect to the uniform measures, and the supremum is taken over all \CPTP~map $\Phi: \B(\Hilbert_d) \to \B(\Hilbert_D)$ (for all finite-dimensional Hilbert spaces $\Hilbert_D$), and all binary \POVM~$\set{M_0, M_1}$. We now follow the same sequence of equations as presented in~\Cref{sec:preliminaries} to derive the result, which is outlined below:
\begin{equation*}
    \Pr[\text{win}] \leq \sup_{M} \operatornorm[\Big]{ \E_{m, k}\ofAlt[\big]{ d \cdot \rho^\T_{m, k} \otimes M_m } },
\end{equation*}
where the supremum is taken over all binary \POVM~$\set{M_0, M_1}$.
We proceed by adapting~\cref{eq:clifford_winning_probability_upperbound}, starting from the observable form $M_i = \frac{I_D + (\shortminus 1)^{i} U}{2}$ for some Hermitian unitary $U$, to derive the following inequalities,
\begin{align*}
    \Pr[\text{win}] &\leq \sup_U \operatornorm[\Bigg]{ \E_{m, k}\ofAlt[\bigg]{ d \cdot \frac{2}{d} \frac{I_d + (\shortminus 1)^m \Gamma_k}{2} \otimes \frac{I_D + (\shortminus 1)^m U}{2} } } \\
    &\leq \sup_U \frac{1}{2 K} \operatornorm[\Bigg]{ \frac{1}{2} \sum_{m, k} \Big( I_d \otimes I_D + (\shortminus 1)^m \of[\big]{\Gamma_k \otimes I_D + I_d \otimes U} + (\shortminus 1)^{2m} \, \Gamma_k \otimes U \Big) } \\
    &\leq \frac{1}{2} + \frac{1}{2 K} \sup_U \operatornorm[\bigg]{\sum_k \Gamma_k \otimes U}\,,
\end{align*}
where the supremum is taken over all Hermitian unitaries $U_k \in \U(\CC^D)$. By applying the sub-multiplicative property of the operator norm (\ie, $\operatornorm{A \otimes B} \leq \operatornorm{A} \cdot \operatornorm{B}$ for all $A$ and $B$), along with the fact that the operator norm of a unitary is one, and the inequality $\operatornorm[\big]{\sum_k \Gamma_k} = \sqrt{K}$, we have the following upper-bound:
\begin{equation*}
    \Pr[\text{win}] \leq \frac{1}{2} + \frac{1}{2 \sqrt{K}}\,.
\end{equation*}
This ensures a quadratic convergence rate for the indistinguishability security of our Clifford encryption scheme.
\section{Analytical and Numerical Results for \texorpdfstring{Conjecture~\ref*{conj:main_conjecture}}{Conjecture 1}} \label{sec:analytical-and-numerical-results}

Although we have been unable to fully prove \Cref{conj:main_conjecture}, we present in this section several analytical and numerical results on the first values of $K \in \NN$. Specifically, we prove the conjecture for $K \leq 7$ (\Cref{sec:Keq2,sec:Kleq7}), provide numerical confirmation for $K \leq 17$ (\Cref{sec:Kleq17}), and present numerical evidence for $K=18$ (\Cref{sec:Keq18}). Additionally, we establish a weaker bound than the conjecture (\Cref{thm:5/8-upper-bound}), which holds for all $K \in \NN$.

\paragraph*{Inapplicability of the Triangular Inequality.}
As we saw in~\Cref{sec:Clifford_Algebra}, a naive triangular inequality of the operator norm in~\cref{eq:clifford_winning_probability_upperbound} yields a trivial upper-bound $\Pr[\text{win}]\leq \sfrac{1}{4} + \sfrac{3K}{4 K}=1$.
In light of the property of~\cref{eq:formula-with-norm-2}, one might consider that a slightly more refined triangular inequality could be sufficient to address~\Cref{conj:main_conjecture}, and that
\begin{equation*}
    \operatornorm[\bigg]{\sum_{k \in [K]} \Gamma_k \otimes \big(U_k \otimes I_D + I_D \otimes U_k\big)} + \operatornorm[\bigg]{\sum_{k \in [K]} I_d \otimes U_k \otimes U_k}\!,
\end{equation*}
would be smaller than $K + 2 \sqrt{K}$, but this is not true for $K > 2$. One should notice that $\operatornorm[\big]{\sum_{k \in [K]} \Gamma_k \otimes U_k} \leq K$, with equality by taking all $U_k = \Gamma_k$. With those $U_k$, the tensor products $U_k \otimes U_k$ are pairwise commuting and thus $\operatornorm[\big]{\sum_{k \in [K]} I_d \otimes U_k \otimes U_k} = K$, but the left-hand side of the triangular inequality $\operatornorm[\big]{\sum_{k \in [K]} \Gamma_k \otimes (U_k \otimes I_D + I_D \otimes U_k)}$ is in general larger than $2 \sqrt{K}$, with first values:
\begin{equation*}
    2 \sqrt{2},\; 2 \sqrt{4},\; 2 \sqrt{6},\; 2 \sqrt{9},\; 2 \sqrt{12},\; 2 \sqrt{16},\; 2 \sqrt{20},\; 2 \sqrt{25},\; \ldots
\end{equation*}
In comparison, when all $U_k = \Gamma_k$, the complete operator norm $\operatornorm[\big]{\sum_{k \in [K]} \Gamma_k \otimes (U_k \otimes I_D + I_D \otimes U_k) + I_d \otimes U_k \otimes U_k}$ is smaller than $K + 2 \sqrt{K}$, with first values:
\begin{equation*}
    4,\; 3,\; 6,\; 7,\; 8,\; 9,\; 10,\; 11,\; \ldots
\end{equation*}

\subsection{Elementary Proofs for \texorpdfstring{$K=2$}{K=2}} \label{sec:Keq2}

We present two distinct preliminary proofs showing that \Cref{conj:main_conjecture} is true for $K = 2$, a first one using the equivalence with the~\cite{BL20} scheme, and another one based on the Sum-of-Squares method. 
In \Cref{sec:Kleq7}, we will extend the second proof to larger values of $K$, with much more material involved. 
Note that a third proof can also be obtained by exploiting the anti-commutation of the two pairs of matrices in $W_2$ and the recent results in \cite{HO22,dGHG23,XSW24,MH24} regarding uncertainty relations. 

    \subsubsection{First Proof, with the BB84 MoE Game}

As we saw in~\Cref{sec:Clifford_Algebra}, for $K=2$ our candidate scheme coincides with the scheme defined in~\cite{BL20}, thus the upper-bound of the winning probability of the no-cloning game for three adversaries $(\Pirate, \Bob, \Charlie)$ is the same as in~\cite{BL20}, \ie $\frac{1}{2} + \frac{1}{2\sqrt2}$. Thus, by~\Cref{prop:three_equalities},
\begin{equation*}
    \Pr[\text{win}] \leq \frac{1}{4} + \frac{1}{4 K} \sup_U \operatornorm[\Big]{\sum_k \Big(\Gamma_k \otimes (U_k \otimes I_D + I_D \otimes U_k) + I_d \otimes U_k \otimes U_k\Big) } = \frac{1}{2} + \frac{1}{2\sqrt2}\,,
\end{equation*}
which implies that $\operatornorm{W_2} \leq 2 + 2 \sqrt{2}$ as claimed.

    \subsubsection{Second Proof, with Sum-of-Squares}
    \label{subsubsec:SOS-for-K-equals-2}

From~\cref{eq:winning_probability_upperbound_2}, the upper-bound of the winning probability of the no-cloning game for three adversaries $(\Pirate, \Bob, \Charlie)$ is
\begin{equation*}
    \Pr[\text{win}] \leq \sup_{\psi, B, C} \bra{\psi} \E_{m, k}\ofAlt[\big]{ d \cdot \rho^\T_{m, k} \otimes B_{m|k} \otimes C_{m|k} } \ket{\psi}\,,
\end{equation*}
where the supremum is taken over all $\norm{\psi} = 1$, all families of \PVM~$\set{B_{i|k}}$ and $\set{C_{j|k}}$, as well as their respective dimensions. Following the same steps as~\Cref{sec:scheme}, we found
\begin{equation} \label{eq:clifford_winning_probability_upperbound_bis}
    \Pr[\text{win}] \leq \frac{1}{4} + \frac{1}{4 K} \sup_{\psi, B, C} \bra{\psi} \sum_k \Big(\Gamma_k \otimes (B_k \otimes I_D + I_D \otimes C_k) + I_d \otimes B_k \otimes C_k \Big)\ket{\psi}\,,
\end{equation}
where the supremum is now taken over all families of observables $\set{B_k}$ and $\set{C_k}$. Note that this time, we do not assume the adversaries $(\Pirate, \Bob, \Charlie)$ to be symmetric, \ie $\Bob$ and $\Charlie$ may have different observables. The last part of the upper-bound~\cref{eq:clifford_winning_probability_upperbound_bis} can be stated as the optimization problem:
\begin{equation} \label{eq:clifford_optimization_problem_tensor}
    \begin{aligned}
        \sup_{\psi, B, C} \quad& \bra{\psi} \sum_k \Big(\Gamma_k \otimes (B_k \otimes I_D + I_D \otimes C_k) + I_d \otimes B_k \otimes C_k \Big) \ket{\psi} \,,\\
        \text{subject to} \quad&\boldsymbol{\cdot}\: \norm{\psi} = 1 \,,\\
        &\boldsymbol{\cdot}\: B^*_i = B_i \quad \boldsymbol{\cdot}\:  C^*_i = C_i \quad \boldsymbol{\cdot}\: B^2 = C^2 = I_D \quad \forall i\,.
    \end{aligned}
\end{equation}
This problem can be relaxed through the use of what is called commuting operator strategies, in which the tensor product structure between Alice's and Bob's operators is replaced by the assumption that these operators commute:
\begin{equation} \label{eq:clifford_optimization_problem_commuting}
    \begin{aligned}
        \sup_{\psi, b, c} \quad& \bra{\psi} \sum_k \Big(\Gamma_k \otimes (b_k + c_k) + I_d \otimes b_k \cdot c_k \Big) \ket{\psi} \,,\\
        \text{subject to} \quad&\boldsymbol{\cdot}\: \norm{\psi} = 1\,, \\
        &\boldsymbol{\cdot}\: b^*_i = b_i \quad \boldsymbol{\cdot}\:  c^*_i = c_i \quad \boldsymbol{\cdot}\: b^2_i = c^2_i = I_D \quad \forall i \,,\\ 
        &\boldsymbol{\cdot}\: \commutator{b_i}{c_j} = 0 \quad \forall i,j\,.
    \end{aligned}
\end{equation}
An optimal value for the problem~\cref{eq:clifford_optimization_problem_tensor} is immediately a lower-bound for the problem~\cref{eq:clifford_optimization_problem_commuting} by taking $b_k \coloneqq B_k \otimes I_D$ and $c_k \coloneqq I_D \otimes C_k$. The question of the equality of the two optimization problems~\cref{eq:clifford_optimization_problem_tensor,eq:clifford_optimization_problem_commuting} (and in general of the tensor-based \emph{versus} the commuting-based models) was a long-standing problem that was only recently refuted by \cite{JNV+21}. However, in the case of finite-dimensional Hilbert spaces, the equality holds as an inductive consequence of Tsirelson's theorem \cite{Tsi93,RXL26}.

If the optimal value of the problem~\cref{eq:clifford_optimization_problem_commuting} is $K + 2 \sqrt{K}$, then the Hermitian operator
\begin{equation}   \label{eq:definition-of-PK}
    P_K \coloneqq (K + 2 \sqrt{K}) \cdot I_d \otimes I_{D^2} - \sum_k \Big(\Gamma_k \otimes (b_k + c_k) + I_d \otimes b_k \cdot c_k\Big)\,,
\end{equation}
must be positive semi-definite under the constraints of \cref{eq:clifford_optimization_problem_commuting}. If we can find a family $\set{H_i}_i$ such that $P_k = \sum_i H_i^* H_i$, then the operator $P_K$ would be guaranteed to be positive semi-definite, as a sum of positive semi-definite operators $H_i^* H_i$. If all $H_i$ are also Hermitian, then we can write $P_K = \sum_i H^2_i$. This constitutes the fundamental principle of the Sum-of-Squares (SoS) decomposition technique, which is used to establish bounds on a variety of quantum correlations, including the CHSH Bell inequality~\cite{Bel64,CHSH69} and its associated Tsirelson bound~\cite{Cir80}.

Let $P(\mathbf{X})$ be an Hermitian polynomial in non-commutative variables $\mathbf{X} = {(X_i)}_i$, then $P(\mathbf{X})$ is a positive semi-definite polynomial (\ie $P(\mathbf{X}) \succeq 0$ for all evaluations of $P$ on matrices $\mathbf{X}$) if and only if $P(\mathbf{X})$ is a Hermitian Sum-of-Squares~\cite{McC01,Hel02}:
\begin{equation*}
    P(\mathbf{X}) = \sum_i H_i(\mathbf{X})^* \; H_i(\mathbf{X})\,.
\end{equation*}
This result can be used to maximize the operator norm of a non-commutative Hermitian polynomial. If the polynomial is constrained to an archimedean semi-algebraic set (a condition that is satisfied in our case), a hierarchy of converging upper-bounds can be obtained based on semi-definite optimization programs (SDP)~\cite{HM04}. The dual problem of those SDPs forms the Navascués-Pironio-Acín (NPA) hierarchy~\cite{NPA08} and correspond to the non-commutative variant of the Lasserre’s hierarchy~\cite{Las01}.

For $K=2$ such SoS decomposition was already known~\cite{BC23}:
\begin{equation*}
    P_2 = \frac{1}{2\sqrt{2}} \of[\big]{ X\otimes b_1 + Z\otimes c_2 - \sqrt{2} \cdot I }^2
    + \frac{1}{2\sqrt{2}} \of[\big]{ X\otimes c_1 + Z\otimes b_2 - \sqrt{2} \cdot I }^2
    + \frac{1}{2} \of[\big]{ b_1 - c_1 }^2
    + \frac{1}{2} \of[\big]{ b_2 - c_2 }^2\,,
\end{equation*}
where we took $\Gamma_1 \coloneqq X$ and $\Gamma_2 \coloneqq Z$. This implies, using~\Cref{prop:three_equalities}, that $\operatornorm{W_2} \leq 2 + 2 \sqrt{2}$.

    \subsection{Proofs for \texorpdfstring{$K\leq7$}{K in \{2, ..., 7\}} and Asymptotic Upper-Bound} \label{sec:Kleq7}

We prove that \Cref{conj:main_conjecture} is true for all $K \leq 7$ in two manners: first using a family of SoS decompositions, then based on its dual SDP problem, the NPA hierarchy.

    \subsubsection{First Proof, with Sum-of-Squares}

We refer to \Cref{subsubsec:SOS-for-K-equals-2} for intuition and background about the SoS method.
Let $K \in \set{2, \ldots, 7}$.
The non-negativity certificate of $P_K$, defined in~\cref{eq:definition-of-PK}, is given by the following SoS:
\begin{equation}\label{eq:first_levels_sos}
    P_K
    \,=\,
    \frac{K-\sqrt{K}}{2K(K-1)} \sum_{i=1}^K \bigg( Q_K + (\sqrt{K}+1)\,\Gamma_i\otimes(c_i-b_i) \bigg)^2
    + 
    \alpha_K\, Q_K^2\,,
\end{equation}
where:
\[
    Q_K
    \,:=
    \,\sqrt{K}\,I\otimes I - \sum_{j=1}^K (\Gamma_j\otimes c_j)
    \qquad\text{and}\qquad
    \alpha_K
    \,:=\,
    \frac{(3K-2)\sqrt{K}-K^2}{2K(K-1)}
    \,.
\]
Notice that the coefficient $\alpha_K$ is positive for all $K\leq 7$, so~\cref{eq:first_levels_sos} is a valide SoS decomposition for $P_2,...,P_7$, which validates~\Cref{conj:main_conjecture} for $K\leq7$. 
However, for $K \geq 8$ the coefficient $\alpha_K$ is negative, so that the decomposition presented in~\cref{eq:first_levels_sos} no longer provides a valid non-negativity certificate. 
Note that it does not exclude the possibility of finding another SoS decomposition valid for larger values of $K$.
It is also worth mentioning that the SoS decomposition given in~\cref{eq:first_levels_sos} can be readily symmetrized with respect to the variables $b_i$ and $c_i$ --- the current formulation has been chosen for its simplicity.

    \subsubsection{Second Proof, with NPA Level-1, and Asymptotic Result}

The NPA hierarchy is dual to the sum-of-squares~(SoS) approach. At level $n$ of the hierarchy, we keep track of only the relations on moments of order up to and including $n$ of the operators. This gives rise to a SDP maximisation whose optimal value is precisely the value that can be certified by a SoS certificate where the squared terms have order $n$. For NPA level $1$, we only consider the relations that are first order in the $b_i$ and $c_i$.

To simplify the relations we work with in the NPA hierarchy, we first express the game value in terms of an anticommuting version of the scenario algebra. Then, at NPA level $1$, we relax the optimisation to only depend on the first-order moments of the observables. 
This turns the problem into an SDP, which can be approximated efficiently. Making use of the symmetry of the problem allows us to drastically simplify the SDP we need to solve, allowing us to do it analytically. 
In \Cref{sec:Kleq17},
we pass to NPA level 2 and we proceed in the same way. 
There, we have to consider the second-order moments. 
We make use of analogous simplifications by symmetry to get a simpler problem; however, this problem becomes too complicated to solve analytically, so we solve it numerically to find the values.

We introduce two scenario algebras and prove that the optimal value of our conjecture is equivalent to the supremum of the operator norm of a new problem: 

\begin{definition}[Scenario Algebra]
    The \emph{scenario algebra} $\mc{A}(K)$ is the $\ast$-algebra generated by $b_1,c_1,\ldots,b_K,c_K$ such that $b_i^2=c_i^2=1$, $b_i^\ast=b_i$, $c_i^\ast=c_i$, and $b_ic_j=c_jb_i$ for all $i,j \in [K]$. 
    
    \noindent The \emph{anticommuting scenario algebra} $\mc{A}_{ac}(K)$ is the $\ast$-algebra generated by $\hat{b}_1,\hat{c}_1,\ldots,\hat{b}_K,\hat{c}_K$ such that $\hat{b}_i^2=\hat{c}_i^2=1$, $\hat{b}_i^\ast=\hat{b}_i$, $\hat{c}_i^\ast=\hat{c}_i$, $\hat{b}_i\hat{c}_i=\hat{c}_i\hat{b}_i$ for all $i$, and $\hat{b}_i\hat{c}_j=-\hat{c}_j\hat{b}_i$ for all $i\neq j \in [K]$.
\end{definition}
In our case, we can take, as a representation, $\Gamma_i\otimes B_i\otimes I$ for $\hat{b}_i$ and $\Gamma_i\otimes I\otimes C_i$ for $\hat{c}_i$. The game polynomial can be seen as an element $p_K\in\mathcal{M}_{d}\otimes\mc{A}(K)$, with $p_K=\sum_k\of[\Big]{\Gamma_k\otimes(b_k+c_k)+I\otimes b_kc_k}$ so that the winning probability is $\frac{1}{4}+\frac{1}{4K}\sup_\pi\norm{(\id\otimes\pi)(p_K)}$, where the supremum is over all finite-dimensional representations of $\mc{A}(K)$.
First, we prove that the winning probabilities over $\mc{A}(K)$ and $\mc{A}_{ac}(K)$ coincide: 

\begin{proposition}
    Let $\hat{p}_K=\sum_k(\hat{b}_k+\hat{c}_k+\hat{b}_k\hat{c}_k)\in\mc{A}_{ac}(K)$. Then, 
    $$\sup_{\pi}\operatornorm{(\id\otimes\pi)(p_K)}=\sup_{\hat{\pi}}\operatornorm{\hat{\pi}(\hat{p}_K)}\,,$$
    where the suprema are over finite-dimensional representations $\pi,\hat\pi$ of $\mc{A}(K)$, $\mc{A}_{ac}(K)$ respectively.
\end{proposition}

\begin{proof}
    Let $\pi$ be a finite-dimensional representation of $\mc{A}(K)$. Then, let $\hat{\pi}$ be the representation of $\mc{A}_{ac}(K)$ defined by $\hat{\pi}(\hat{b}_k)=\Gamma_k\otimes\pi(b_k)$ and $\hat{\pi}(\hat{c}_k)=\Gamma_k\otimes\pi(c_k)$. It is direct to see that this satisfies the relations of $\mc{A}_{ac}(K)$ and $(\id\otimes\pi)(p_K)=\hat{\pi}(\hat{p}_K)$, and hence that $\operatornorm{(\id\otimes\pi)(p_K)}=\operatornorm{\hat{\pi}(\hat{p}_K)}$. Taking suprema, we get $\sup_{\pi}\operatornorm{(\id\otimes\pi)(p_K)}\leq\sup_{\hat{\pi}}\operatornorm{\hat{\pi}(\hat{p}_K)}$.

    For the other direction, let $\hat{\pi}$ be a finite-dimensional representation of $\mc{A}_{ac}(K)$. Then, define a representation $\pi$ of $\mc{A}(K)$ by $\pi(b_k)=\Gamma_k^\T\otimes\hat{\pi}(\hat{b}_k)$ and $\pi(c_k)=\Gamma_k^\T\otimes\hat{\pi}(\hat{c}_k)$. As above, it is direct to see that this satisfies the relations of $\mc{A}(K)$. Let $\ket{\psi}$ be a unit vector such that $\operatornorm{\hat{\pi}(\hat{p}_K)}=\abs*{\braket{\psi}{\hat{\pi}(\hat{p}_K)}{\psi}}$, and write $\ket{\Psi_d}\in\C^d\otimes\C^d$ for the maximally entangled state. Then, $\braket{\Psi_d}{\Gamma_k\otimes\Gamma_k^\T}{\Psi_d}=1$, so
    \begin{align*}
        \operatornorm{(\id\otimes\pi)(p_K)}
        &\,\geq\,
        \abs[\Big]{\big\langle{\Psi_d\otimes\psi}\,\big|\,{\sum_k\Big(\Gamma_k\otimes\Gamma_k^\T\otimes(\hat{\pi}(\hat{b}_k)+\hat{\pi}(\hat{c}_k))+I_{d^2}\otimes\hat{\pi}(\hat{b}_k\hat{c}_k)\Big)}\,\big|\,{\Psi_d\otimes\psi}\big\rangle}
        \\
        &\,=\,
        \abs[\big]{\braket{\psi}{\hat{\pi}(\hat{p}_K)}{\psi}}
        \,=\,
        \operatornorm[\big]{\hat{\pi}(\hat{p}_K)}\,,
    \end{align*}
    giving the inequality $\sup_{\pi}\operatornorm{(\id\otimes\pi)(p_K)}\geq\sup_{\hat{\pi}}\operatornorm{\hat{\pi}(\hat{p}_K)}$.
\end{proof}

Thus, the winning probability can be found via an optimization over representations of $\mc{A}_{ac}(K)$. Let the bias of a strategy be $4K\Pr[\text{win}]-K$. Then the optimal bias is $\beta_K=\sup_{\hat{\pi}}\norm{\hat{\pi}(\hat{p}_K)}$ and it is the solution of the following optimization:
\begin{equation*}
    \begin{aligned}
        \sup_{\psi, b, c} \quad& \bra{\psi} \sum_k \Big(b_k + c_k + b_k \cdot c_k\Big) \ket{\psi}\,, \\
        \text{subject to} \quad&\boldsymbol{\cdot}\: \norm{\psi} = 1\,, \\
        &\boldsymbol{\cdot}\: b^*_i = b_i \quad \boldsymbol{\cdot}\:  c^*_i = c_i \quad \boldsymbol{\cdot}\: b^2_i = c^2_i = I_D \quad \forall i\,, \\ 
        &\boldsymbol{\cdot}\: b_ic_i=c_ib_i \quad \forall i\,, \\ 
        &\boldsymbol{\cdot}\: b_ic_j=-c_jb_i \quad \forall i\neq j\,.
    \end{aligned}
\end{equation*}
By our conjecture, the value of the optimal bias should be $K+2\sqrt{K}$. We consider the first level of the NPA hierarchy on this algebra~\cite{NPA08}. To do so, we write $\ket{u_i}=b_i\ket{\psi}$, and $\ket{v_i}=c_i\ket{\psi}$, and dilate the parameter spaces so that the only relations on these vectors are those that follow directly from the relations on the operators. For example, we have $\braket{u_i}{v_j}=\braket{\psi}{b_ic_j}{\psi}=-\braket{\psi}{c_jb_i}{\psi}=-\braket{v_j}{u_i}$ for all $i\neq j$. We have the level-$1$ NPA optimization:
\begin{align}
    \label{eq:NPA-level-1-optimization-problem}
\begin{split}
	\sup_{\psi, u, v} \quad&\sum_{i}\Big({\braket{\psi}{u_i}+\braket{\psi}{v_i}+\braket{u_i}{v_i}}\Big)\,,\\
	\text{subject to} \quad&\boldsymbol{\cdot}\:\braket{\psi}{\psi}=\braket{u_i}{u_i}=\braket{v_i}{v_i}=1\quad \forall\,i\,,\\
	&\boldsymbol{\cdot}\:\braket{\psi}{u_i}=\braket{u_i}{\psi} \quad \;\boldsymbol{\cdot}\:\braket{\psi}{v_i}=\braket{v_i}{\psi}\quad \forall\,i\,,\\
	&\boldsymbol{\cdot}\:\braket{u_i}{v_i}=\braket{v_i}{u_i}\quad \forall\,i\,,\\
	&\boldsymbol{\cdot}\:\braket{u_i}{v_j}=-\braket{v_j}{u_i} \quad \forall\,i\neq j\,.
\end{split}
\end{align}
Then, we use the Gram matrix~$G$ of the vectors
$\of[\big]{\ket{\psi},\ket{u_1},\ldots,\ket{u_K},\ket{v_1},\ldots,\ket{v_K}}$ to reformulate the optimization problem:
\begin{equation*}
    G =
    \begin{pmatrix}
        \braket{\psi}{\psi} & \braket{\psi}{u_1} & \dots & \braket{\psi}{u_K} & \braket{\psi}{v_1} & \dots & \braket{\psi}{v_K} \\
        \braket{u_1}{\psi} & \braket{u_1}{u_1} & \dots & \braket{u_1}{u_K} & \braket{u_1}{v_1} & \dots & \braket{u_1}{v_K} \\
        \vdots & \vdots &  & \vdots & \vdots & & \vdots \\
        \braket{u_K}{\psi} & \braket{u_K}{u_1} & \dots & \braket{u_K}{u_K} & \braket{u_K}{v_1} & \dots & \braket{u_K}{v_K} \\
        \braket{v_1}{\psi} & \braket{v_1}{u_1} & \dots & \braket{v_1}{u_K} & \braket{v_1}{v_1} & \dots & \braket{v_1}{v_K} \\
        \vdots & \vdots &  & \vdots & \vdots & & \vdots \\
        \braket{v_K}{\psi} & \braket{v_K}{u_1} & \dots & \braket{v_K}{u_K} & \braket{v_K}{v_1} & \dots & \braket{v_K}{v_K} \\
    \end{pmatrix}\,.
\end{equation*}
The first constraint in \cref{eq:NPA-level-1-optimization-problem} yields that the diagonal elements of~$G$ are all~$1$. Also, since~$G$ is positive semidefinite, the second and third constraints specify elements of~$G$ that are real; and the fourth constraint specifies elements of~$G$ that are imaginary.
Therefore, taking:
\begin{equation*}
    H = \frac{1}{2}
    \begin{pmatrix}
        0 & \bra{1_K} & \bra{1_K} \\
        \ket{1_K} & 0 & I \\
        \ket{1_K} & I & 0
    \end{pmatrix}\,,
\end{equation*}
where $\ket{1_K}\in\C^K$ is the column vector of ones, the optimization becomes:
\begin{align}
\begin{split}\label{eq:gram_matrix_npa_1}
	\sup_{G} \quad&\Tr(HG)\,,\\
	\text{subject to} \quad&\boldsymbol{\cdot}\:G(\psi,\psi)=G(u_i,u_i)=G(v_i,v_i)=1\quad \forall\,i\,,\\
	&\boldsymbol{\cdot}\:G(\psi,u_i), G(\psi,v_i), G(u_i,v_i)\in\R\quad \forall\,i\,,\\
	&\boldsymbol{\cdot}\:G(u_i,v_j)\in i\R\quad\forall\,i\neq j\,,\\
	&\boldsymbol{\cdot}\:G\succeq 0\,.
\end{split}
\end{align}
We can reduce this problem to only three variables:

\begin{lemma}\label{lem:simplified-sdp-npa-1}
    The optimal value of the SDP \cref{eq:gram_matrix_npa_1} is equal to the value of the optimization
    \begin{align*}
    	\sup_{g} \quad&2Kg_1+Kg_3\,,\\
    	\text{subject to} \quad&\boldsymbol{\cdot}\:\begin{pmatrix}1&g_1\bra{1_K}&g_1\bra{1_K}\\g_1\ket{1_K}&I+g_2(\ketbra{1_K}{1_K}-I)&g_3I\\g_1\ket{1_K}&g_3I&I+g_2(\ketbra{1_K}{1_K}-I)\end{pmatrix}\succeq 0\,,\\
        &\boldsymbol{\cdot}\:g_1,g_2,g_3\in\R\,.
\end{align*}
\end{lemma}

\begin{proof}
    The simplification exploits the symmetries of $H$, which can be used to induce symmetries on $G$. First, note that $H$ has real components, and therefore $G$ can be assumed to be real, yielding the optimization
    \begin{align*}
    	\sup_{G} \quad&\Tr(HG)\,,\\
    	\text{subject to} \quad&\boldsymbol{\cdot}\:G(\psi,\psi)=G(u_i,u_i)=G(v_i,v_i)=1\quad\forall i\,,\\
    	&\boldsymbol{\cdot}\:G(u_i,v_j)=0\quad\forall\,i\neq j\,,\\
    	&\boldsymbol{\cdot}\:G\succeq 0\,.
    \end{align*}
    Next, for any permutation $\sigma\in \mathfrak S_K$, the matrix $H$ is invariant under the permutation of the indices $\of[\big]{u_i\mapsto u_{\sigma(i)}, v_j\mapsto v_{\sigma(j)}}$. 
    Moreover, $H$ is invariant under the permutation $\of[\big]{u_i\mapsto v_i, v_i\mapsto u_i}$. 
    Thus, $G$ can be supposed to be invariant under those permutations, and we can add the following constraints to the optimization: 
    $G(\psi,u_i)=G(\psi,u_j)=G(\psi,v_i)=G(\psi,v_j)=:g_1$ 
    and $G(u_i,v_i)=G(u_j,v_j)=:g_3$ for all $i,j$, and $G(u_i,u_j)=G(u_k,u_l)=G(v_i,v_j)=G(v_k,v_l)=:g_2$ for all $i\neq j$, $k\neq l$. These additional constraints take the SDP to the form given in the statement.
\end{proof}

By the previous lemma, we can find the optimal value by first finding the set of feasible points $(g_1,g_2,g_3)$ explicitly, and then finding the optimum of the objective function $2Kg_1+Kg_3$ over that set.

\begin{lemma}\label{lem:eigenvalues-1}
	Let $A=\sum_ia_i\ketbra{i}{i}$ and $B=\sum_ib_i\ketbra{i}{i}$ be commuting Hermitian matrices. Then we have the  eigendecomposition
	$$\begin{pmatrix}A&B\\B&A\end{pmatrix}\,=\,\sum_i\of[\Big]{(a_i+b_i)\ketbra{i+}{i+}+(a_i-b_i)\ketbra{i-}{i-}}\,,$$
	where $\ket{i+}=\frac{1}{\sqrt{2}}\begin{psmallmatrix}\ket{i}\\\ket{i}\end{psmallmatrix}$ and $\ket{i-}=\frac{1}{\sqrt{2}}\begin{psmallmatrix}\ket{i}\\ \shortminus\ket{i}\end{psmallmatrix}$.
\end{lemma}

\begin{proof}
	We can express $\begin{psmallmatrix}A&B\\B&A\end{psmallmatrix}=\sum_i\begin{psmallmatrix}a_i&b_i\\b_i&a_i\end{psmallmatrix}\otimes\ketbra{i}{i}$. Now, $\begin{psmallmatrix}a_i&b_i\\b_i&a_i\end{psmallmatrix}=(a_i+b_i)\ketbra{+}{+}+(a_i-b_i)\ketbra{-}{-}$, giving the result.
\end{proof}

\begin{lemma}\label{lem:eigenvalues-2}
	Let $H=\sum_{i}\lambda_i\ketbra{v_i}{v_i}$ be a Hermitian matrix. Then the Hermitian matrix $\begin{psmallmatrix}1&\omega\bra{v_1}\\\omega\ket{v_1}&H\end{psmallmatrix}$ has eigenvalues $\lambda_i$ for $i>1$, and $\frac{1+\lambda_1}{2}\pm\sqrt{\of*{\frac{1-\lambda_1}{2}}^2+\omega^2}$.
\end{lemma}

\begin{proof}
	Noting that $\begin{psmallmatrix}0\\\ket{v_i}\end{psmallmatrix}$ is an eigenvector with eigenvalue $\lambda_i$ for all $i>1$, the remaining two eigenvalues are the eigenvalues of $\begin{psmallmatrix}1&\omega\\\omega&\lambda_1\end{psmallmatrix}$, which are of the above form.
\end{proof}

\begin{theorem}  \label{thm:5/8-upper-bound}
    The value of the first level of the NPA hierarchy for the Clifford uncloneable encryption MoE game is $\frac{1}{2}+\frac{1}{2\sqrt{K}}$ for $K\leq 7$, and $\frac{5}{8}+\frac{1}{2(K-2)}-\frac{1}{4K}$ for $K>7$.
    In particular, we obtain the upper-bound of $5/8=0.625$ on the winning probability of the no-cloning game in the limit $K\rightarrow\infty$.
\end{theorem}

\begin{proof}
    First, we want to find the feasible points of the optimization in \Cref{lem:simplified-sdp-npa-1} by calculating the eigenvalues of the matrices $G$ of that form. Using \Cref{lem:eigenvalues-1}, we see that the eigenvalues of $$\begin{pmatrix}I+g_2\,(\ketbra{1_K}{1_K}\shortminus I)&g_3\,I\\g_3\,I&I+g_2(\ketbra{1_K}{1_K}\shortminus I)\end{pmatrix}$$ are $Kg_2+(1-g_2)+g_3$ with eigenvector $\frac{1}{\sqrt{2K}}\ket{1_{2K}}$, and $Kg_2+(1-g_2)-g_3$, $1-g_2+g_3$, $1-g_2-g_3$ with eigenvectors orthogonal to $\frac{1}{\sqrt{2K}}\ket{1_{2K}}$. Using \Cref{lem:eigenvalues-2}, the eigenvalues of $G$ are $Kg_2+(1-g_2)-g_3$, $1-g_2+g_3$, $1-g_2-g_3$ and $\frac{1}{2}(1+Kg_2+(1\shortminus g_2)+g_3)\pm\sqrt{\frac{1}{4}\of*{1\shortminus (Kg_2+(1\shortminus g_2)+g_3)}^2+2Kg_1^2}$. In order for this to be a feasible point of the SDP, all of these eigenvalues must be positive. This means that $1+(K-1)g_2\geq g_3$, $1-g_2\geq\pm g_3$, and $1+(K-1)g_2+g_3\geq 2Kg_1^2$.

    To find the optimal point, consider a new parametrisation $x=1-g_2+g_3$, $y=1-g_2-g_3$, $\lambda=2g_1+g_3$. Then, $g_1=\lambda/2-(x-y)/4$, $g_2=1-(x+y)/2$, $g_3=(x-y)/2$ so the objective function becomes $K\lambda$ and the constraints become $x,y\geq 0$, $2\geq x+\frac{K\shortminus2}{K}\,y$, and $\lambda^2-(x-y)\lambda+\frac{(x\shortminus y)^2}{4}+\frac{K\shortminus2}{K}\,x+y-2\leq 0$, which is a degree-$2$ polynomial in $\lambda$. 
    As $\lambda$ is to be maximised, we get:
    $$
        \lambda
        \,=\,
        \frac{x\shortminus y}{2}+\sqrt{2-\frac{K\shortminus2}{K}\,x-y}\,\,.
    $$
    Since $\lambda$ decreases with $y$, it is maximised at $y=0$, giving $x\in[0,2]$ and $\lambda=\frac{x}{2}+\sqrt{2-\frac{K\shortminus2}{K}\,x}$. 
    Its derivative is:
    $$
        \frac{d\lambda}{dx}
        \,=\,
        \frac{1}{2}-\frac{K\shortminus2}{2K\sqrt{2-\frac{K\shortminus2}{K}\,x}}\,\,,
    $$
    which is $0$ if and only if $x=\frac{2K}{K\shortminus2}-\frac{K\shortminus2}{K}$, and positive for smaller $x$. If $K\leq 7$, then this value of $x$ is greater than $2$, so $\lambda$ is maximised at $x=2$, giving
    $\lambda=1+\frac{2}{\sqrt{K}},$
    and hence winning probability 
    $$
        w_1
        \,=\,
        \frac{1}{4}+\frac{\lambda}{4}
        \,=\,
        \frac{1}{2}+\frac{1}{2\sqrt{K}}\,\,.
    $$
    
    If $K\geq 8$, the maximum value is attained in the interval $[0,2]$ so the optimum
    $\lambda=\frac{K}{K\shortminus2}+\frac{K\shortminus2}{2K},$
    and therefore the winning probability 
    \[
        w_1
        \,=\,
        \frac{5}{8}+\frac{1}{2(K-2)}-\frac{1}{4K}\,\,. \qedhere
    \]
\end{proof}

\subsection{Numerical Result with NPA Hierarchy Level 2 for \texorpdfstring{$K\leq17$}{K in \{2, ..., 17\}}} \label{sec:Kleq17}

To get a better upper-bound on the game, we make use of a higher level of the NPA hierarchy. Here, we use level~$2$ of the NPA hierarchy, where we optimize over a Gram matrices indexed by quadratic terms in the generators of $\mc{A}_{ac}(K)$. Even using the symmetries of the problem, as in~\Cref{lem:simplified-sdp-npa-1}, this becomes a vastly more difficult optimization, and requires us to do it numerically. However, we observe the same behaviour as for NPA level 1: the optimal value of the SDP matches the conjectured value exactly, until a certain $K$, where it starts to diverge towards a higher limit than $1/2$. Here, the point of divergence is $K=18$. We give some of the optimal values for the NPA level 1 and 2 optimizations in the table below, highlighting the points of divergence:
\begin{center}
    \begin{tabular}{|c|c|c|c|}
        \hline
        $K$ & NPA level $1$ & NPA level $2$ & Conjecture\\\hline
        2&0.8536&0.8536&0.8536\\
        4&0.7500&0.7500&0.7500\\
        7&0.6890&0.6890&0.6890\\
        8&0.6771&0.6768&0.6768\\
        12&0.6542&0.6443&0.6443\\
        16&0.6451&0.6250&0.6250\\
        17&0.6436&0.6213&0.6213\\
        18&0.6424&0.6182&0.6179\\
        25&0.6367&0.6062&0.6000\\
        35&0.6330&0.5980&0.5845\\\hline
    \end{tabular}
\end{center}
Note also that the NPA level 2 upper-bound of $w_2=0.5980$ for $K=35$ is, to the best of our knowledge, the best known unconditional upper-bound on the security of an uncloneable encryption scheme.

To finish this section, we outline the construction of the SDP that allows us to find the NPA level 2 optimal values more efficiently. In particular, it allows us to reduce the number of free parameters from $1+3K^2$ to $18$. First, the second level of the NPA hierarchy can be derived similarly to the first one by considering the optimization of the value over vectors $\ket{u_i}=b_i\ket{\psi}$, $\ket{v_i}=c_i\ket{\psi}$, $\ket{u_iv_j}=b_ic_j\ket{\psi}$ for all $i,j$, and $\ket{u_iu_j}=b_ib_j\ket{\psi}$, $\ket{v_iv_j}=c_ic_j\ket{\psi}$ for $i\neq j$ (noting that $\ket{u_iu_i}=\ket{v_iv_i}=\ket{\psi}$), and taking the relations on the vectors to be those that follow immediately from the relations on the operators. As for the first level, we may assume that the Gram matrix is real, and that it is invariant under the permutations of the indices $i\mapsto\sigma(i)$ and $u_i\leftrightarrow v_i$ as in \Cref{lem:simplified-sdp-npa-1}. Then, this gives us that the optimization is over $18$ independent parameters $g_1,\ldots,g_{18}\in\R$, where the inner products of the vectors satisfy
\begin{align*}
    1&=\braket{\psi}{\psi}=\braket{u_i}{u_i}=\braket{v_i}{v_i}=\braket{u_iv_i}{u_iv_i}=\braket{u_iv_j}{u_iv_j}=\braket{u_iu_j}{u_iu_j}=\braket{v_iv_j}{v_iv_j}\\
    0&=\braket{u_i}{v_j}=\braket{\psi}{u_iv_j}=\braket{u_iv_i}{u_iu_j}=\braket{u_iv_i}{u_ju_i}=\braket{u_iv_j}{u_iu_k}=\braket{u_iv_j}{u_ku_i}\\
    &=\braket{u_iv_j}{u_ku_l}=\braket{u_iv_i}{v_iv_j}=\braket{u_iv_i}{v_jv_i}=\braket{u_iv_j}{v_kv_j}=\braket{u_iv_j}{v_jv_k}=\braket{u_iv_j}{v_kv_l}\\
    &=\braket{u_iu_j}{v_kv_i}=\braket{u_iu_j}{v_jv_k}\\
    g_1&=\braket{\psi}{u_i}=\braket{\psi}{v_i}=\braket{u_i}{u_iv_i}=\braket{v_i}{u_iv_i}=\braket{u_i}{u_iv_j}=-\braket{v_i}{u_jv_i}=\braket{u_i}{u_iu_j}
    % \\&
    =\braket{v_i}{v_iv_j}\\
    g_2&=\braket{u_i}{v_i}=\braket{\psi}{u_iv_i}=\braket{u_iv_j}{u_iu_j}=-\braket{u_iv_j}{v_jv_i}\\
    g_3&=\braket{u_i}{u_j}=\braket{v_i}{v_j}=\braket{v_iv_j}{v_iv_j}=\braket{u_iv_i}{u_iv_j}=-\braket{u_iv_i}{u_jv_i}=\braket{\psi}{u_iu_j}=\braket{\psi}{v_iv_j}\\
    &=\braket{u_iv_j}{u_iv_k}=\braket{u_iv_j}{u_kv_j}=\braket{u_iu_j}{u_iu_k}=\braket{v_iv_j}{v_iv_k}\\
    g_4&=\braket{u_i}{u_jv_j}=-\braket{u_i}{u_jv_i}=\braket{v_i}{u_jv_j}=\braket{v_i}{u_iv_j}=\braket{v_i}{u_iu_j}=-\braket{v_i}{u_ju_i}\\
    &=\braket{u_i}{v_iv_j}=-\braket{u_i}{v_jv_i}\\
    g_5&=\braket{u_i}{u_jv_k}=-\braket{v_i}{u_jv_k}=\braket{v_i}{u_ju_k}=\braket{u_i}{v_jv_k}\\
    g_6&=\braket{u_iv_i}{u_jv_j}=-\braket{u_iu_j}{v_jv_i}\\
    g_7&=\braket{u_iv_j}{u_jv_i}=-\braket{u_iu_j}{v_iv_j}\\
    g_8&=\braket{u_iv_j}{u_kv_i}=-\braket{u_iu_j}{v_iv_k}=\braket{u_iu_j}{u_kv_k}\\
    g_9&=\braket{u_iv_j}{u_kv_l}=\braket{u_iu_j}{v_kv_l}\\
    g_{10}&=\braket{u_i}{u_ju_i}=\braket{v_i}{v_jv_i}\\
    g_{11}&=\braket{u_i}{u_ju_k}=\braket{v_i}{v_jv_k}\\
    g_{12}&=\braket{u_iv_j}{u_ju_i}=-\braket{u_iv_j}{v_iv_j}\\
    g_{13}&=\braket{u_iv_i}{u_ju_k}=\braket{u_iv_j}{u_ku_j}=\braket{u_iv_i}{v_jv_k}=-\braket{u_iv_j}{v_kv_i}\\
    g_{14}&=\braket{u_iv_j}{u_ju_k}=-\braket{u_iv_j}{v_iv_k}\\
    g_{15}&=\braket{u_iu_j}{u_ju_i}=\braket{v_iv_j}{v_jv_i}\\
    g_{16}&=\braket{u_iu_j}{u_ku_i}=\braket{v_iv_j}{v_kv_i}\\
    g_{17}&=\braket{u_iu_j}{u_ku_j}=\braket{v_iv_j}{v_kv_j}\\
    g_{18}&=\braket{u_iu_j}{u_ku_l}=\braket{v_iv_j}{v_kv_l}\\
\end{align*}
for all distinct $1\leq i,j,k,l\leq K$. Then, the SDP simplifies to the form
\begin{align*}
    	\sup_{g} \quad&2\,K\,g_1+K\,g_2\,,\\
    	\text{subject to} \quad&\boldsymbol{\cdot}\:G_0+g_1\,G_1+\ldots+g_{18}\,G_{18}\succeq 0\,,\\
        &\boldsymbol{\cdot}\:g_1,g_2,\ldots,g_{18}\in\R\,,
\end{align*}
where $G_0$ is the $1+3K^2$-dimensional matrix that has components $1$ where the Gram matrix $G$ is $1$ and zeros elsewhere, $G_1$ is the matrix that has components $\pm1$ where $G$ is $\pm g_1$, $G_2$ is the matrix that has components $\pm1$ where $G$ is $\pm g_2$, and so on. We solve this SDP numerically to find the winning probabilities displayed in the second column of the above table.

\subsection{Numerical Result with the Alternating Optimization Algorithm} \label{sec:Keq18}

\begin{figure}
    \centering
    \includegraphics[width=0.4\linewidth]{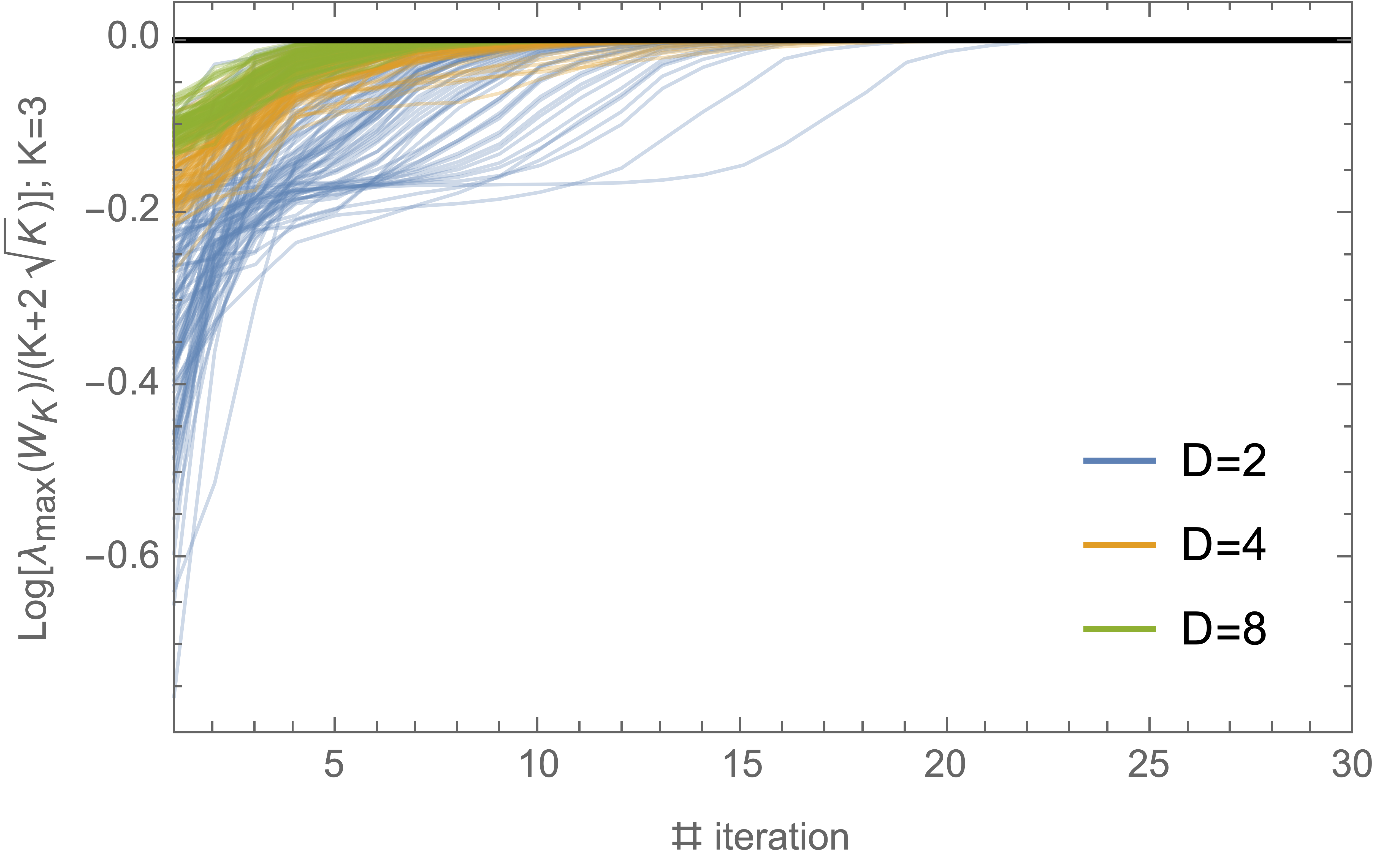}
    \includegraphics[width=0.4\linewidth]{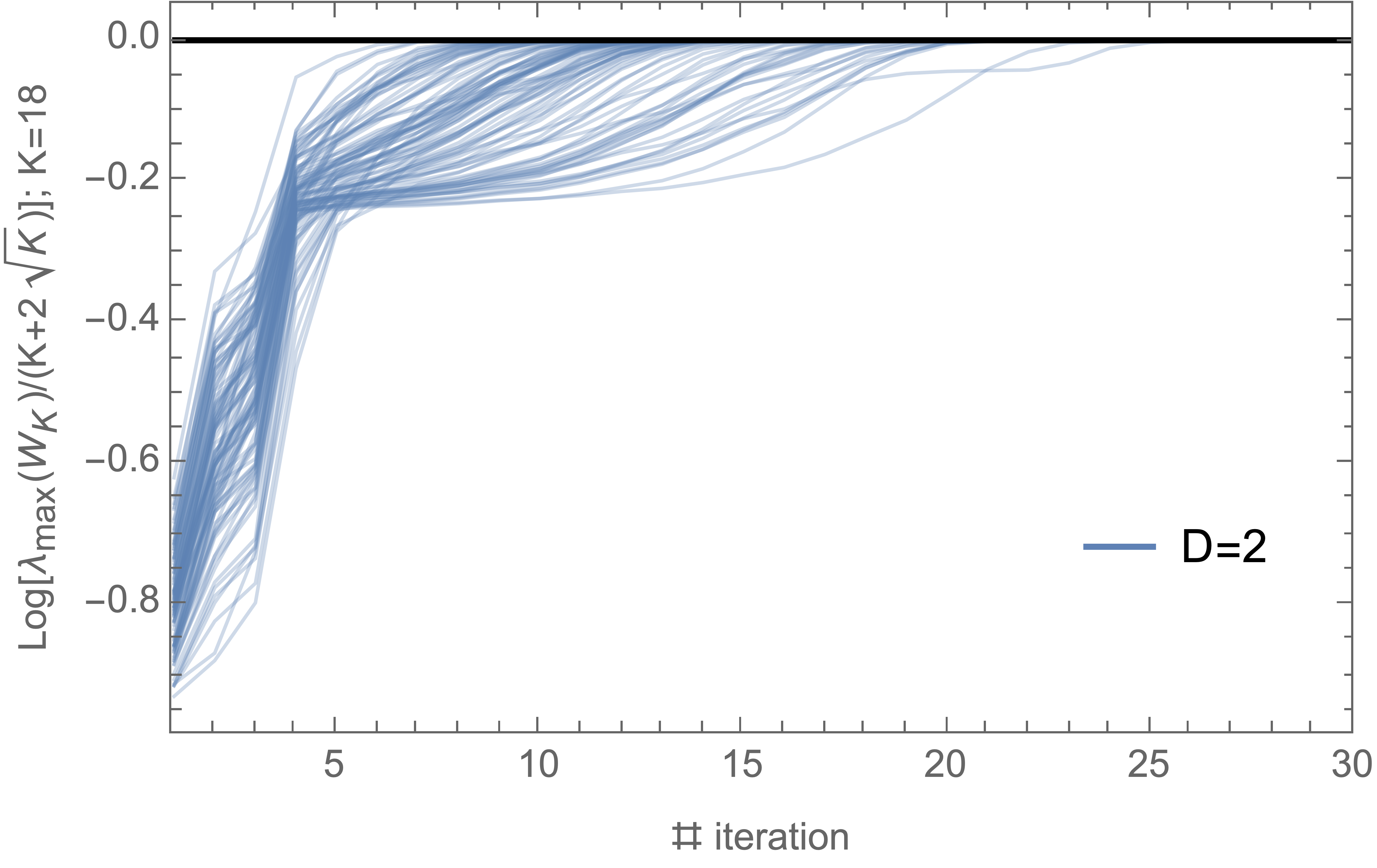}
    \caption{Optimizing the maximum eigenvalue of the operator $W_K$ using a seesaw algorithm. Left panel: $K=3$; right panel: $K=18$. We consider different matrix sizes $D$ and we run $100$ instances of the algorithm (with random initializations) for each dimension, for $M=10$ iterations of the three steps. The~$X$ axis tracks the intermediate step of the algorithm, ranging from $1$ to $3M=30$. The~$Y$ axis tracks the log-relative error $\log[\lambda_{\max}(W_K)/(K + 2\sqrt K)]$. Note that all the curves are below 0, providing evidence towards \Cref{conj:main_conjecture}.}
    \label{fig:seesaw}
\end{figure}

We present in this subsection a heuristic algorithm for optimizing the largest eigenvalue of the operator
$$
W_K = \sum_{k=1}^K \Big(\Gamma_k \otimes B_k \otimes I + \Gamma_k \otimes I \otimes C_k + I \otimes B_k \otimes C_k\Big)\,,
$$
over self-adjoint contractions $B_1, C_1, \ldots, B_K, C_K$. We shall use the \emph{alternating optimization} (or \emph{seesaw}) method~\cite{BH02}. Note that our goal is to maximize a convex function over a convex set, \ie
\[
    \max_{\substack{\|z\| = 1\\\operatornorm{B_k} \leq 1 \\ \operatornorm{C_k} \leq 1}} \langle z | W_K | z \rangle\,,
\]
hence our problem does not fall in the classical convex optimization framework. 

Our algorithm optimizes iteratively over each of the variables $z$, $\{B_k\}_{k=1}^K$, and $\{C_k\}_{k=1}^K$ for a pre-determined number of iterations $M$. The variables $z, B_k, C_k$ are initialized with random values ($z$ uniform on the unit sphere of $\CC^d \otimes \CC^D \otimes \CC^D$, and $B_k$, $C_k$ i.i.d. with Haar-distributed eigenvectors and half of eigenvalues $\pm 1$). We give below the details of each optimization step. 
\begin{enumerate}
    \item Optimizing $z$: compute the largest eigenvalue of the current operator $P$ and assign to $z$ the corresponding eigenvector. 
    \item Optimizing the contractions $B$: permute the first two tensors to rewrite the problem as 
    $$\langle z | \sum_{k=1}^K \Gamma_k \otimes I \otimes C_k | z \rangle + \max_{\operatornorm{B_k} \leq 1} \langle z_{1 \leftrightarrow 2} | \sum_{k=1}^K B_k \otimes (\Gamma_k \otimes I + I \otimes C_k)| z_{1 \leftrightarrow 2} \rangle\,,$$
    where $z_{1 \leftrightarrow 2}$ is the 3-tensor $z$ with the first two factors permuted. Clearly, the maximum above is equal to the sum of maxima over individual $B_k$'s that we can further decompose as 
    \begin{align*}
        &\max_{\operatornorm{B_k} \leq 1} \langle z_{1 \leftrightarrow 2} |  B_k \otimes (\Gamma_k \otimes I + I \otimes C_k)| z_{1 \leftrightarrow 2} \rangle \\ &= \max_{\operatornorm{B_k} \leq 1} \langle B_k , Z_{13|2} ^*(\Gamma_k \otimes I + I \otimes C_k) Z_{13|2} \rangle \\
        &= \|Z_{13|2}^* (\Gamma_k \otimes I + I \otimes C_k) Z_{13|2} \|_1\,,
    \end{align*}
    where $Z_{13|2} \in \mathcal M_{Dd \times D}$ is the reshaping of the 3-tensor $z_{1 \leftrightarrow 2}$, see~\Cref{fig:z-diagrams}. In the last equality above we have used the following fact: 
    $$\max_{\operatornorm{B} \leq 1} \langle B, X \rangle = \|X\|_1\,,$$
    with the maximum being attained for 
    $$B_{\textrm{opt}} = \sum_i \operatorname{sign}(\lambda_i) \ketbra{x_i}{x_i} \qquad \text{for} \qquad X = \sum_i \lambda_i \ketbra{x_i}{x_i}\,.$$
    We apply this procedure for all the maximization problems corresponding to the $B_k$'s and update the matrices $B_k$ accordingly. 
    \item Optimizing the contractions $C$: similar procedure as above, up to tensor permutation. 
\end{enumerate}

\begin{figure}[t]
    \centering
    \includegraphics[width=0.7\linewidth]{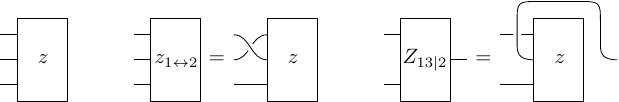}
    \caption{Graphical representations of the tensors $z$, $z_{1 \leftrightarrow 2}$, and $Z_{13|2}$.}
    \label{fig:z-diagrams}
\end{figure}

We present in~\Cref{fig:seesaw} the results of numerical experiments in the case $K=3$ and $K=18$ for different matrix sizes $D$. The results we find agree with \Cref{conj:main_conjecture}.

\subsection*{Acknowledgements}

The authors gratefully acknowledge Arthur Mehta for the many useful comments and support throughout the development of this work. 
We are also thankful to Sang-Jung Park for the insightful discussions and constructive feedback on the preliminary version of the manuscript. Special thanks go to Benoît Collins for his many enlightening discussions, as well as to Andreas Bluhm, Mariia Elovenkova, Aabhas Gulati, and Faedi Loulidi for their thoughtful comments and suggestions.

We acknowledge the support of the Natural Sciences and Engineering Research Council of Canada (NSERC), MITACS grant FR113029, and of the ANR project \href{https://esquisses.math.cnrs.fr/}{ESQuisses}, grant number ANR-20-CE47-0014-01.
P.B.~acknowledges the support of the Institute for Quantum Technologies in Occitanie. 
E.C.~acknowledges the support of a CGS D scholarship from Canada's NSERC. 
C.P.~acknowledges the support of the ANR project “Quantum Trajectories” grant no. ANR-20-CE40-0024-0.
D.R.~acknowledges the support of the Air Force Office of Scientific Research under award number FA9550-20-1-0375.
This research was enabled in part by support provided by  \href{https://www.computeontario.ca/}{Compute Ontario} and the \href{https://alliancecan.ca}{Digital Research Alliance of Canada}. 

\small
\bibliographystyle{bibtex/bst/alphaarxiv.bst}
\bibliography{
  bibtex/bib/quasar-full.bib,
  bibtex/bib/quasar.bib,
  bibtex/bib/quasar-more.bib
  }

\end{document}